\documentclass{elsarticle}
\usepackage[
  paper = a4paper,
  left = 1.0in,
  top = 1.0in,
  right = 1.0in,
  bottom = 1.0in
]{geometry}
\usepackage{amsthm, amsmath, amssymb, amsfonts, amsbsy}
\usepackage{mathrsfs} 
\usepackage{esint} 
\allowdisplaybreaks
\usepackage{indentfirst}
\usepackage{graphicx}
\usepackage{epstopdf} 
\usepackage{makeidx}
  \makeindex
  
\usepackage{pdfpages} 
\usepackage{caption}
\usepackage{subcaption}
\usepackage{longtable}
\usepackage{xcolor}
  \definecolor{b+}{rgb}{0.5, 0.5, 1}
  \definecolor{b-}{rgb}{0, 0, 0.5}
  \definecolor{r|b-}{rgb}{0.5, 0, 0.25}
  \definecolor{g|b-}{rgb}{0, 0.5, 0.25}
  \definecolor{shadecolor}{rgb}{0.75, 0.75, 1}
\usepackage[
  colorlinks,
  urlcolor = black,
  linkcolor = b-,
  anchorcolor = r|b-,
  citecolor = g|b-
]{hyperref}

\renewcommand{\(}{\left(}
\renewcommand{\)}{\right)}
\newcommand{\lsb}{\left[}
\newcommand{\rsb}{\right]}
\newcommand{\lbk}{\left\{}
\newcommand{\rbk}{\right\}}

\newcommand{\lnm}{\left\|}
\newcommand{\rnm}{\right\|}
\newcommand{\lmdl}{\left|}
\newcommand{\rmdl}{\right|}

\newcommand{\rme}{\mathrm{e}}

\newcommand{\tp}{^{T}} 

\newcommand{\interior}[1]{\mathrm{Int} \( {#1} \)} 
\newcommand{\bigO}[1]{\mathcal{O} \( {#1} \)} 
\newcommand{\tr}[1]{\mathrm{Tr} \( {#1} \)} 
\newcommand{\uvec}[1]{\hat{\mathbf{#1}}} 

\DeclareMathOperator{\atan2}{atan2}
\DeclareMathOperator{\esssup}{esssup}

\theoremstyle{definition}
\newtheorem{deft}{Definition}[section]

\theoremstyle{definition}
\newtheorem{nt}{Notation}[section]

\theoremstyle{definition}
\newtheorem{asp}{Assumption}[section]

\theoremstyle{definition}
\newtheorem{thm}{Theorem}[section]

\theoremstyle{definition}
\newtheorem{lemma}{Lemma}[section]

\theoremstyle{definition}
\newtheorem{coro}{Corollary}[section]

\theoremstyle{definition}
\newtheorem{prop}{Proposition}[section]

\theoremstyle{definition}
\newtheorem{eg}{Example}[section]

\theoremstyle{definition}
\newtheorem{ex}{Exercise}[section]

\theoremstyle{remark}
\newtheorem{rmk}{Remark}[section]

\theoremstyle{plain}
\newtheorem{rmkplain}{Remark}[section]

\numberwithin{equation}{section}
\numberwithin{figure}{section}
\numberwithin{table}{section}

\begin{document}

\title{Solution Landscapes in the Landau-de Gennes Theory on Rectangles}
\author[1]{Lidong Fang}
\ead{ldfang.sjtu@gmail.com}
\author[2]{Apala Majumdar}
\ead{a.majumdar@bath.ac.uk}
\author[1]{Lei Zhang}
\ead{lzhang2012@sjtu.edu.cn}
\address[1]{School of Mathematical Sciences, Institute of Natural Sciences, and MOE-LSC, Shanghai Jiao Tong University, 800 Dongchuan Road, Shanghai 200240, China.}
\address[2]{Department of Mathematics and Statistics, University of Strathclyde, Glasgow, G11XH.}

\begin{abstract}
  We study nematic equilibria on rectangular domains, in a reduced two-dimensional Landau-de Gennes framework. These reduced equilibria carry over to the three-dimensional framework at a special temperature. There is one essential model variable---$\epsilon$ which is a geometry-dependent and material-dependent variable. We compute the limiting profiles exactly in two distinguished limits---the $\epsilon \to 0$ limit relevant for macroscopic domains and the $\epsilon \to \infty$ limit relevant for nano-scale domains. The limiting profile has line defects near the shorter edges in the $\epsilon \to \infty$ limit whereas we observe fractional point defects in the $\epsilon \to 0$ limit. The analytical studies are complemented by bifurcation diagrams for these reduced equilibria as a function of $\epsilon$ and the rectangular aspect ratio. We also introduce the concept of `non-trivial' topologies and the relaxation of non-trivial topologies to trivial topologies mediated via point and line defects, with potential consequences for non-equilibrium phenomena and switching dynamics. 
\end{abstract}
\maketitle

\tableofcontents



\section{Introduction}

Nematic liquid crystals are quintessential examples of partially ordered materials intermediate between solids and liquids \cite{degennes1995physics} \cite{stewart2004static}. The constituent nematic molecules are typically asymmetric in shape and the macroscopic nematic state has a degree of long-range orientational ordering, i.e., the constituent molecules tend to align along locally preferred directions referred to as `directors' in the literature. The directional nature of nematics makes them highly responsive materials and indeed, their sensitivity to light and external electric fields have made nematics the working material of choice for the multi-billion dollar liquid crystal display industry \cite{bahadur1990liquid}. There is huge contemporary interest in understanding pattern formation in nematic systems, in the context of micro-patterned systems, thin geometries including surfaces and interfacial phenomena \cite{gallardo1999electrochemical}, \cite{brake2003biomolecular}, \cite{lagerwall2014cellulose}, \cite{lagerwall2012new}. 

There are several continuum macroscopic theories for nematic liquid crystals. We focus on the celebrated Landau-de Gennes (LdG) theory for nematics, in a reduced two-dimensional (2D) setting, as described in \autoref{sec:theory}. This approach is accepted for very thin domains, which are treated as two-dimensional domains and the nematic directors are assumed to be in the plane of the domain with tangent or in-plane boundary conditions. This reduction can be rigorously justified (see \cite{golovaty2017dimension}) and for a special temperature, we can make an immediate connection between the reduced 2D approach and the full three-dimensional (3D) approach, as explained in detail in the next section. The 2D LdG theory is essentially the Ginzburg-Landau theory \cite{bethuel1994ginzburg}, phrased in terms of the 2D LdG $\mathbf{Q}$-tensor, which is a symmetric and traceless $2\times 2$ matrix with two degrees of freedom. We refer to the eigenvector with the largest positive eigenvalue as being the nematic director in this framework and the order parameter, $s$, is proportional to the corresponding eigenvalue. The corresponding equilibria are local or global minimizers of a reduced 2D LdG free energy, which is effectively the Ginzburg-Landau free energy with one rescaled parameter, $\epsilon$. This parameter encodes the geometric length scale of the domain, material properties and the nematic elasticity; more precisely at the fixed temperature under consideration, it can be interpreted as the ratio of the nematic correlation length to the macroscopic domain size \cite{kralj2014order}. The nematic correlation length is a characteristic material and temperature-dependent length scale related to the size of nematic defects \cite{kralj2010finite}.

The prototype problem of nematics inside square domains has been studied extensively \cite{luo2012multistability}, \cite{kusumaatmaja2015free}, \cite{lewis2014colloidal}, \cite{kralj2014order}, following the experimental and modelling work presented in \cite{tsakonas2007multistable}. The reduced 2D LdG equilibria on square domains with tangent boundary conditions, which require the nematic director to be tangent or parallel to the square edges, have been well classified. In particular, in \cite{kralj2014order}, the authors report a new Well Order reconstruction Solution (WORS) in this reduced setting, distinguished by a pair of mutually orthogonal defect lines along the square diagonals. The WORS is the unique nematic equilibrium for sufficiently large values of $\epsilon$. For small $\epsilon$, there are at least six different competing equilibria : two `diagonal' states for which the director is aligned along one of the square diagonals and four `rotated' states, for which the director rotates by $\pi$ radians between a pair of opposite edges. It is very natural to study the counterpart of this problem on rectangles with tangent boundary conditions, to systematically study the effects of geometrical anisotropy. 

It is difficult to make analytic progress for the problem in general, except for two distinguished limits: the $\epsilon \to \infty$ limit which describes nano-scale geometries or `small' geometries, and the $\epsilon \to 0$ limit which describes macroscopic domains with size much greater than the nematic correlation length. In the $\epsilon \to \infty$ limit, the problem effectively reduces to a Dirichlet boundary-value problem for both the components of the limiting profile, $\mathbf{Q} _{\infty}$.
We show that the cross structure of the WORS is lost on all rectangles.
However, the cross structure of the WORS is retained on square domains, if we replace the Dirichlet conditions with surface energies that enforce tangent conditions on the square edges, for large values of the surface anchoring parameter, $W \propto \epsilon$ for large values of $\epsilon$. In other words, the cross structure is specific to a square domain but not an artefact of Dirichlet conditions. On a rectangle, the limiting profile, $\mathbf{Q} _{\infty}$ has two nodal lines or defect lines along the shorter rectangular edges, referred to as sBD2 solutions in the remainder of the manuscript. The sBD2 solutions also survive on rectangles with surface energies and large values of the surface anchoring parameter i.e. the surface energies do not change the qualitative conclusions for large values of $W$. We use a combination of formal calculations and elegant maximum principle arguments to illustrate the effects of geometrical anisotropy in this limit.

In the $\epsilon \to 0$ limit, the limiting problem is just the Laplace equation for an angle $\theta$ in the plane of the rectangle, subject to Dirichlet conditions on the rectangular edges. In this limit, $s$ is a constant away from the rectangular vertices and the director, $\uvec{n} = \( \cos \theta, \sin \theta \)$. The tangent conditions necessarily create discontinuities in $\theta$ at the vertices and the limiting energy is effectively the Dirichlet energy of $\theta$ away from the vertices (also see \cite{majumdar2010landau}, \cite{frank1958liquid} and \cite{bethuel1993asymptotics}). There are three competing nematic equilibria in this limit---the diagonal solutions, the rotated states (R1, R2) for which the director rotates between a pair of horizontal edges and the rotated (R3, R4) solutions for which the director rotates between a pair of vertical edges. The four rotated states are not energetically degenerate on a rectangle, in contrast to a square. We use simple symmetry-based arguments to compare the respective energies of these solutions and these arguments may be of wider interest. Using standard arc-continuation methods, we numerically investigate the solution landscapes on rectangles as a function of $\epsilon$ and geometrical anisotropy. The bifurcation diagrams are qualitatively similar to those reported for a square (see \cite{robinson2017molecular}) except for the fact that the energetically expensive R1, R2 solutions are disconnected from the connected diagonal and R3, R4 solution branches. The geometrical anisotropy affects R1, R2 solution branches, in the sense that they cannot be numerically continued to larger values of $\epsilon$ and our numerical methods cannot track the pathway from the R1, R2 solutions to the unique sBD2 solution as $\epsilon$ increases. This might have interesting experimental consequences, particularly for transition states that connect the R1, R2 stable states to the diagonal or R3, R4 states. We will investigate this in future work. 

We introduce the concept of non-trivial topologies, based on the profile of $\theta$ near the rectangular vertices. This is largely a theoretical notion, inspired by super-twisted nematic devices in a planar setting \cite{bahadur1990liquid}. The director is constrained on the rectangular edges by virtue of the tangent conditions but is free to rotate in the plane, between the edges. The minimal allowed rotation is labelled as `trivial' and excess rotation is labelled as `non-trivial'. We speculate that non-trivial topologies can be realised by localised rotation of the nematic molecules near the vertices of a rectangular domain. Analogous remarks apply to shallow 3D wells with a rectangular cross-section. These non-trivial topologies will relax to trivial topologies when the rotation is removed. The relaxation pathways are mediated by defects of different dimensionality, depending on $\epsilon$ and the geometrical anisotropy, and can offer new optical possibilities or switching mechanisms between trivial states e.g. diagonal, R3 and R4 states. The relaxation pathways will be even more interesting with external electric fields, which are not accounted for in this manuscript.

To summarise, our paper can be regarded as a set of results for reduced LdG equilibria on rectangular domains with tangent boundary conditions. These results highlight the effects of geometrical anisotropy and to some extent, boundary conditions in certain limiting cases. The interplay between $\epsilon$ and the geometrical anisotropy has not come across yet. Such studies do not necessarily have unexpected outcomes but are much needed for systematic and structured investigations of the effects of geometry on nematic equilibria e.g. regular polygons versus polygons with sides of unequal length, dimensions of defect sets and control of defect sets and multiplicity of equilibria by tuning the geometrical parameters.

The paper is organised as follows. In \autoref{sec:theory}, we review the reduced 2D LdG theory, the free energy, the rescaling, the governing partial differential equations, and the boundary conditions. In \autoref{sec:limits}, we analytically study distinguished asymptotic limits on a rectangular domain with tangent boundary conditions. In \autoref{sec:bifurcation}, we numerically study how the 2D LdG equilibria depend on $\epsilon$, for different values of the rectangular anisotropy. In \autoref{sec:nontrivial}, we discuss relaxation mechanisms and their dependence on $\epsilon$ and we conclude in \autoref{sec:conclusions} with some perspectives.

\section{Theoretical Framework}
\label{sec:theory}

We work within the continuum Landau-de Gennes (LdG) theory for nematic liquid crystals. The LdG theory is one of the most powerful continuum theories for nematics in the literature \cite{degennes1995physics}. It describes the state of nematic anisotropy by the LdG $\mathbf{Q} _{\textrm{LdG}}$-tensor order parameter---a symmetric, traceless $3 \times 3$ matrix whose eigenvectors model the directions of averaged molecular alignment in space (referred to as directors) and the eigenvalues are a measure of the degree of orientational order about the corresponding eigenvectors. From the spectral decomposition theorem, we can write $\mathbf{Q} _{\textrm{LdG}}$ as
\begin{align}
  \mathbf{Q} _{\textrm{LdG}} = \sum _{i = 1} ^{3} \lambda _{i} \uvec{e} _{i} \otimes \uvec{e} _{i},
\end{align}
where $\lbk \uvec{e} _{1}, \uvec{e} _{2}, \uvec{e} _{3} \rbk$ constitute an orthonormal basis and $\sum _{i = 1} ^{3} \lambda _{i} = 0$. A $\mathbf{Q} _{\textrm{LdG}}$-tensor is said to be (i) isotropic if $\mathbf{Q} _{\textrm{LdG}} = 0$, (ii) uniaxial if $\mathbf{Q} _{\textrm{LdG}}$ has a pair of degenerate non-zero eigenvalues and (iii) biaxial if $\mathbf{Q} _{\textrm{LdG}}$ has three distinct eigenvalues \cite{degennes1995physics}. We refer to the `director' as being the eigenvector with the largest positive eigenvalue.

We study nematic equilibria on three-dimensional wells with a rectangular cross section, i.e., our domain is
\begin{align}
  \label{eq:domain3D}
  \mathcal{B} = \lbk \( x, y, z \) \in \mathbb{R} ^{3} \mid \( x, y \) \in \Omega _{aL, bL}, 0 \leq z \leq h \rbk,
\end{align}
where $h > 0$ is the height of the well, and for $p, q > 0$,
\begin{align}
  \Omega _{p, q} := \lbk \( x, y \) \in \mathbb{R} ^{2} \mid 0 \leq x \leq p, 0 \leq y \leq q \rbk.
\end{align}
Here, $L > 0$ is a fixed length, $a$ and $b$ measure the geometrical anisotropy and in what follows, we always fix $b = 1$ and vary $a$. For the sake of simplicity, we will drop the subscripts in $\Omega _{aL, bL}$.

We impose tangent boundary conditions on the well surfaces, in the form of Dirichlet uniaxial conditions on the lateral surfaces and surface energies on the top and bottom surfaces, that enforce planar degenerate anchoring, i.e., constrain the director to be in the plane of the surface whilst respecting the tangent conditions on the edges of the rectangular cross-section. Our work builds on the substantial work done for square domains  \cite{tsakonas2007multistable}, \cite{luo2012multistability}, \cite{kralj2014order}, \cite{robinson2017molecular}, \cite{wang2019order} and more recently on rectangular domains \cite{walton2018nematic}.

A particularly simple form of the LdG free energy is given by
\begin{align}
  F _{\textrm{LdG}} \lsb \mathbf{Q} _{\textrm{LdG}} \rsb 
  := \int _{\mathcal{B}} \( \frac{K}{2} \lmdl \nabla \mathbf{Q} _{\textrm{LdG}} \rmdl ^{2} + f _{\textrm{B}} \( \mathbf{Q} _{\textrm{LdG}} \) \) + \int _{\partial \mathcal{B}} f _{\textrm{S}} \( \mathbf{Q} _{\textrm{LdG}} \),
\end{align}
where $K > 0$ is a material-dependent elastic constant, $f _{\textrm{B}}$ is the bulk energy, $f _{\textrm{S}}$ is the surface energy, and
\begin{align}
  \lmdl \nabla \mathbf{Q} _{\textrm{LdG}} \rmdl ^{2} 
  & := \sum _{i, j, k = 1} ^{3} \( \partial _{k} \( Q _{\textrm{LdG}} \) _{ij} \) ^{2}, \\
  \label{eq:fB}
  f _{\textrm{B}} \( \mathbf{Q} _{\textrm{LdG}} \) 
  & := \frac{A}{2} \tr{\mathbf{Q} _{\textrm{LdG}} ^{2}} - \frac{B}{3} \tr{\mathbf{Q} _{\textrm{LdG}} ^{3}} + \frac{C}{4} \( \tr{\mathbf{Q} _{\textrm{LdG}} ^{2}} \) ^{2}.
\end{align}
The variable $A = \alpha \( T - T ^{*} \)$ is a rescaled temperature, $\alpha, B, C > 0$ are material-dependent constants, and $T ^{*}$ is the characteristic nematic supercooling temperature. Further, $\tr{\mathbf{Q} _{\textrm{LdG}} ^{2}} = \sum _{i, j = 1} ^{3} \( Q _{\textrm{LdG}} \) _{ij} \( Q _{\textrm{LdG}} \) _{ji}$ and $\tr{\mathbf{Q} _{\textrm{LdG}} ^{3}} = \sum _{i, j, k = 1} ^{3} \( Q _{\textrm{LdG}} \) _{ij} \( Q _{\textrm{LdG}} \) _{jk} \( Q _{\textrm{LdG}} \) _{ki}$.
The rescaled temperature $A$ has three characteristic values \cite{majumdar2010equilibrium}: 
(i) $A = 0$, below which the isotropic phase $\mathbf{Q} _{\textrm{LdG}} \equiv 0$ loses stability, 
(ii) the nematic-isotropic transition temperature, $A = \frac{B ^{2}}{27C}$, at which $f _{\textrm{B}}$ is minimized by the isotropic phase and a continuum of uniaxial states, 
and (iii) the nematic superheating temperature, $A = \frac{B ^{2}}{24C}$, above which the isotropic state is the unique critical point of $f _{\textrm{B}}$.
We work with $A < 0$ throughout this manuscript for which $f _{\textrm{B}}$ favours an ordered uniaxial nematic phase  and the physically relevant/observable nematic equilibria correspond to either global or local LdG minimizers subject to the imposed boundary conditions.

Following the methods in \cite{golovaty2015dimension} and for certain physically relevant choices of the surface anchoring energies (that favour planar degenerate anchoring on the top and bottom surfaces) on the top and bottom surfaces, we can rigorously justify the reduction from the three-dimensional domain $\mathcal{B}$ to the two-dimensional domain $\Omega$ , in the $h \to 0$ limit or the thin film limit.
In this limit, the physically relevant equilibria belong to the class of $\mathbf{Q} _{\textrm{LdG}}$-tensors that have the unit vector in the $z$-direction as a fixed eigenvector so that $\mathbf{Q} _{\textrm{LdG}}$ can be written as
\begin{align}
  \label{eq:Q_fixed_z}
  \mathbf{Q} _{\textrm{LdG}}
  = \( 
  \begin{array}{ccc} 
    q _{1} - q _{3} & q _{2} & 0 \\ 
    q _{2} & -q _{1} - q _{3} & 0 \\ 
    0 & 0 & 2 q _{3} 
  \end{array} 
  \),
\end{align}
where the three degrees of freedom of $\mathbf{Q} _{\textrm{LdG}}$ are contained in the three scalar functions---$q _{1}$, $q _{2}$ and $q _{3}$ which only depend on the spatial variables in the cross-section. The details of this reduction can be found in \cite{wang2019order}. 
In this limit, it suffices to study minimizers of the LdG free energy on the two-dimensional cross-section $\Omega$ in the class of $\mathbf{Q} _{\textrm{LdG}}$-tensors in \autoref{eq:Q_fixed_z}. Moreover, in \cite{canevari2019well}, the authors compute rigorous bounds for $q _{3}$ as a function of $A$, building on previous work \cite{ignat2016instability} in a different context. They show that for a particular choice $A = -\frac{B ^{2}}{3C}$, all competing equilibria have constant $q _{3} \equiv -\frac{B}{6C}$ so that the physically relevant $\mathbf{Q} _{\textrm{LdG}}$-tensors can be written as
\begin{align} 
  \label{eq:Q_3D_2D} 
  \mathbf{Q} _{\textrm{LdG}} 
  = \(
  \begin{array}{cc}
    \mathbf{Q} & \\
    & 0
  \end{array}
  \)
  - \frac{B}{6C} \(
  \begin{array}{ccc} 
    -1 & & \\ 
    & -1 & \\ 
    & & 2 
  \end{array} 
  \)
\end{align} 
with only two degrees of freedom. Here, $\mathbf{Q}$ is a symmetric traceless $2\times 2$ matrix with only two degrees of freedom that account for the nematic ordering in the plane of $\Omega$.
Substituting \autoref{eq:Q_3D_2D} into the LdG free energy and using Gamma-convergence methods in the $h \to 0$ limit as in \cite{golovaty2015dimension}, one can show that for $A = -\frac{B ^{2}}{3C}$, the LdG free energy minimizers converge uniformly to minimizers of the reduced energy (after suitable rescaling $\bar{\mathbf{Q}} \( x, y \) = \mathbf{Q} \( xL, yL \)$)
\begin{align}
  \label{eq:LdG_2D}  
  \frac{1}{h K} F \lsb \bar{\mathbf{Q}} \rsb
  := \int _{\tilde{\Omega}} \( \frac{1}{2} \lmdl \nabla \bar{\mathbf{Q}} \rmdl ^{2} + \frac{L ^{2}}{K} \( -\frac{B ^{2}}{4C} \tr{\bar{\mathbf{Q}} ^{2}} + \frac{C}{4} \( \tr{\bar{\mathbf{Q}} ^{2}} \) ^{2} \) - \frac{5}{432} \frac{B ^{4} L ^{2}}{C ^{3} K} \),
\end{align}
and
\begin{align}
  \label{eq:domain2D}
  \tilde{\Omega} := \Omega _{a, b} = \lsb 0, a \rsb \times \lsb 0, b \rsb.
\end{align}
In \autoref{eq:LdG_2D}, we have reduced the original three-dimensional problem to a two-dimensional problem on the re-scaled rectangle $\tilde{\Omega}$ at the fixed temperature $A = -\frac{B^2}{3C}$ and the remainder of this manuscript is devoted to this reduced problem; generalizations will be discussed in the Conclusions. 
We can equivalently write $\bar{\mathbf{Q}}$ in terms of a scalar order parameter field $s: \tilde{\Omega} \to \mathbb{R}$ and a nematic director $\uvec{n}: \tilde{\Omega} \to \mathbb{S} ^{1}$, which is a unit-vector field that models the preferred direction of nematic alignment in the plane of $\tilde{\Omega}$ as shown below. For $\( x, y \) \in \tilde{\Omega}$,
\begin{align}
  \label{eq:Q_s_n}
  \bar{\mathbf{Q}} \( x, y \) = 2 s \( x, y \) \( \uvec{n} \( x, y \) \otimes \uvec{n} \( x, y \) - \mathbf{I} _{2} / 2 \right)
\end{align}
Alternatively,
\begin{align}
  \bar{\mathbf{Q}} \( x, y \) =
  \(
  \begin{array}{cc}
    \bar{Q} _{11} \( x, y \) &  \bar{Q} _{12} \( x, y \) \\
    \bar{Q} _{12} \( x, y \) & -\bar{Q} _{11} \( x, y \)
  \end{array}
  \),
\end{align}
where we can relate the matrix components $\bar{Q} _{11}$ and $\bar{Q} _{12}$ to $s$ and $\uvec{n}$ by defining a new variable $\theta := \atan2 \( \bar{Q} _{12}, \bar{Q} _{11} \) / 2$, where,
\begin{align}
  \atan2 \( y, x \) = \left\{
  \begin{array}{ll}
    \arctan \( y / x \), & x > 0, \\
    \arctan \( y / x \) + \pi, & x \leq 0, y \geq 0, \\
    \arctan \( y / x \) - \pi, & x \leq 0, y < 0.
  \end{array}
  \right.
\end{align}
When $\bar{\mathbf{Q}} \neq 0$, $\theta$ is well-defined, and we have the following important relations,
\begin{align}
  \label{eq:Q_s_n_componentwise}
  \uvec{n} \( x, y \) & = \( \cos \theta \( x, y \), \sin \theta \( x, y \) \) \tp, \\
  \bar{Q} _{11} \( x, y \) & = s \( x, y \) \cos \( 2 \theta \( x, y \) \), \\
  \bar{Q} _{12} \( x, y \) & = s \( x, y \) \sin \( 2 \theta \( x, y \) \).
\end{align}

We rescale the reduced energy \autoref{eq:LdG_2D} by defining $s _{0} := \frac{B}{2C}$ (so that $\tr{\bar{\mathbf{Q}} ^{2}} = 2 s _{0} ^{2}$ for minimizers of the bulk potential in \autoref{eq:LdG_2D}) and $\tilde{\mathbf{Q}}$ as
\begin{align}
  \label{eq:Qtilde}
  \tilde{\mathbf{Q}} \( x, y \) := \bar{\mathbf{Q}} \( x, y \) / s _{0}, \quad \( x, y \) \in \tilde{\Omega},
\end{align}
along with
\begin{align}
  \label{eq:epsilon}
  \epsilon := \frac{\sqrt{2 C K}}{B L}.
\end{align} 
Since we have fixed $A = -\frac{B ^{2}}{3C}$, the ratio $\xi := \sqrt{\frac{KC}{B ^{2}}}$ is proportional to the nematic correlation length \cite{virga1995variational}, which is a characteristic material-dependent length scale. Therefore, $\epsilon$ is the ratio of the nematic correlation length to the domain size (measured in terms of $L$).
The rescaled energy is then equivalent to
\begin{align}
  \label{eq:LdG_2D_dimless}
  \frac{4 C ^{2}}{B ^{2} h K} \tilde{F} \lsb \tilde{\mathbf{Q}} \rsb 
  := \int_{\tilde{\Omega}} \( \lmdl \nabla \tilde{Q} _{11} \rmdl ^{2} + \lmdl \nabla \tilde{Q} _{12} \rmdl ^{2} + \frac{1}{\epsilon ^{2}} \( \tilde{Q} _{11} ^{2} + \tilde{Q} _{12} ^{2} - 1 \) ^{2} - \frac{32}{27} \frac{1}{\epsilon ^{2}} \),
\end{align}
and all critical points (including global and local energy minimizers) are classical solutions of the associated Euler-Lagrange equations, for smooth boundary conditions.
The Euler-Lagrange equations are given by the following system of elliptic coupled partial differential equations, also referred to as the Ginzburg-Landau equations on $\tilde{\Omega}$ in the literature \cite{bethuel1993asymptotics}.
\begin{align}
  \label{eq:EL_stronganchoring}
  \left\{
  \begin{array}{c}
    \Delta \tilde{Q} _{11} = \epsilon ^{-2} \( \tilde{Q} _{11} ^{2} + \tilde{Q} _{12} ^{2} - 1 \) \tilde{Q} _{11},
    \\
    \Delta \tilde{Q} _{12} = \epsilon ^{-2} \( \tilde{Q} _{11} ^{2} + \tilde{Q} _{12} ^{2} - 1 \) \tilde{Q} _{12}.
  \end{array}
  \right.
\end{align}

The last ingredient is the choice of boundary conditions. As in \cite{tsakonas2007multistable}, \cite{luo2012multistability}, we work with tangent boundary conditions so that $\theta$ is a multiple of $\pi$ on the horizontal edges (defined by $y = 0$ and $y = b$) and $\theta$ is an odd multiple of $\pi / 2$ on the vertical edges (defined by $x = 0$ and $x = a$). Recalling the relations \autoref{eq:Q_s_n_componentwise}, this implies that $\tilde{Q} _{12} = 0$ on all four edges, $\tilde{Q} _{11} = +1$ on the horizontal edges and $\tilde{Q} _{11} = -1$ on the vertical edges.
However, there is a necessary mismatch at the corners; we adopt the same interpolatory approach as in \cite{luo2012multistability} and define a vector field $\tilde{\mathbf{g}} _{d}$ for $0 < d < \min \lbk a / 2, b / 2 \rbk$ as,
\begin{align}
  \label{eq:boundary_g}
  \tilde{\mathbf{g}} _{d; a, b} \( x, y \) :=
  \left\{
  \begin{array}{cc}
    \lsb + T _{\frac{d}{a}} \( \frac{x}{a} \), 0 \rsb, & x \in \lsb 0, a \rsb, y \in \lbk 0, b \rbk, \\
    \lsb - T _{\frac{d}{b}} \( \frac{y}{b} \), 0 \rsb, & y \in \lsb 0, b \rsb, x \in \lbk 0, a \rbk.
  \end{array}
  \right.
\end{align}
The trapezoidal function $T _{d} : \lsb 0, 1 \rsb \to \mathbb{R}$ is given by,
\begin{align}
  \label{eq:boundary_T}
  T _{d} \( t \) = \min \lbk t / d, 1, \( 1 - t \) / d \rbk = \left\{
  \begin{array}{cc}
    t / d, & 0 \leq t \leq d, \\
    1, & d \leq t \leq 1 - d, \\
    \( 1 - t \) / d, & 1 - d \leq t \leq 1.
  \end{array}
  \right.
\end{align}
For the sake of simplicity, we will use the brief notation $\tilde{\mathbf{g}} _{d}$ instead of $\tilde{\mathbf{g}} _{d; a, b}$.

We then fix $\( \tilde{Q} _{11}, \tilde{Q} _{12} \) = \tilde{\mathbf{g}} _{d} $ on $\partial \tilde{\Omega}$. We do not rigorously justify this choice of boundary conditions but we believe that the qualitative trends are not affected by the choice of $d$ above, provided it is sufficiently small. In what follows, we analytically and numerically study the reduced and rescaled tensor $\tilde{\mathbf{Q}} \in W ^{1, 2} \( \tilde{\Omega}; S _{0} \)$ subject to the Dirichlet condition, where $S _{0}$ is the space of symmetric, traceless $2 \times 2$ matrices and $W ^{1, 2}$ is the usual Sobolev space \cite{evans2010partial}.

\section{Two Limiting Problems in terms of \texorpdfstring{$\epsilon$}{epsilon}}
\label{sec:limits}

We can make analytic progress with the solutions of the system \autoref{eq:EL_stronganchoring} in two distinguished limits, the $\epsilon \to \infty$ and $\epsilon \to 0$ limits respectively. The two limits have different physical interpretations. For $A = -\frac{B ^{2}}{3C}$, the length scale $\xi = \sqrt{\frac{KC}{B ^{2}}}$ is proportional to the nematic correlation length, which is typically of the order of tens to hundreds of nanometres and describes small-scale phenomena \cite{virga1995variational}. Hence, the first limit, $\epsilon \to \infty$, is relevant for domains with characteristic length comparable to the nematic correlation length, e.g., few hundred nanometres \cite{kralj2014order}, \cite{canevari2017order}. The second limit, $\epsilon \to 0$, is relevant for domains with characteristic length $L$ much larger than the nematic correlation length, e.g., $L$ is on the micron-scale which is the usual experimentally achievable length scale although nano-scale geometries are becoming a reality with advanced nano-fabrication techniques.

\subsection {The \texorpdfstring{$\epsilon \to \infty$}{epsilon Goes to Infinity} Limit}
\label{sec:limits_Laplace}

\subsubsection{Strong Anchoring Condition}
\label{sec:strong}

In this section, we discuss solutions of \autoref{eq:EL_stronganchoring} in the $\epsilon \to \infty$ limit, with the Dirichlet condition $\( \tilde{Q} _{11}, \tilde{Q} _{12} \) = \tilde{\mathbf{g}} _{d}$ defined in \autoref{eq:boundary_g}. Dirichlet conditions are referred to as `strong anchoring' in the literature. In practice, it is experimentally difficult to realise Dirichlet conditions but Dirichlet problems are relatively analytically tractable and are often, a good approximation to realistic scenarios.

One can prove there exists a critical $\epsilon _{c} > 0$ such that \autoref{eq:EL_stronganchoring} has a unique solution for $\epsilon > \epsilon _{c}$ \cite{lamy2014bifurcation}, \cite{canevari2017order}. The limit of large $\epsilon$ with strong anchoring has been studied extensively on square domains, where the authors of \cite{kralj2014order} and \cite{canevari2017order} report the new Well Order Reconstruction Solution (WORS) in this limit. The reduced WORS $\bar{\mathbf{Q}}$-tensor in \autoref{eq:LdG_2D} vanishes along the square diagonals and the nodal lines of $\bar{\mathbf{Q}}$ partition the square into four quadrants such that the director $\uvec{n}$ is constant in each quadrant. In \cite{canevari2017order}, the authors study the WORS in terms of solutions of the scalar Allen-Cahn equation at the special temperature $A = -\frac{B ^{2}}{3C}$ and rigorously prove that the WORS exists for all $\epsilon > 0$, is globally stable for $\epsilon$ large enough and loses stability as $\epsilon$ decreases.

A natural question is---is the WORS specific to a square domain with strong anchoring conditions in the $\epsilon \to \infty$ limit? We address this question in two parts: (i) by studying solutions of \autoref{eq:EL_stronganchoring} on a rectangle with $a > b = 1$ in $\Omega$ and strong anchoring conditions and (ii) by studying solutions of \autoref{eq:EL_stronganchoring}, in the $\epsilon \to \infty$ limit, on a square domain with weak anchoring, i.e., a weaker implementation of the Dirichlet condition \autoref{eq:boundary_g} as described below. 

In the strong anchoring case, it is relatively straightforward to prove that $\lnm \tilde{\mathbf{Q}} \rnm _{L ^{\infty} \( \tilde{\Omega} \)}$ is bounded independently of $\epsilon$ for all solutions of \autoref{eq:EL_stronganchoring} by maximum principle arguments as in \cite{majumdar2010equilibrium} (here $L ^{\infty}$ is the usual norm i.e. $\lnm \tilde{\mathbf{Q}} \rnm _{L ^{\infty} \( \tilde{\Omega} \)} = \esssup _{\mathbf{r} \in \tilde{\Omega}} \lnm \tilde{\mathbf{Q}} \( \mathbf{r} \) \rnm$). Hence, in the limit $\epsilon \to \infty$, the coupled system \autoref{eq:EL_stronganchoring} reduces to the uncoupled Laplace equations on $\tilde{\Omega}$ as given below, 
\begin{align}
  \label{eq:EL_limit_infty}
  \left\{
  \begin{array}{c}
    \Delta \tilde{Q} _{11} = 0,
    \\
    \Delta \tilde{Q} _{12} = 0.
  \end{array}
  \right.
\end{align}

Our first result shows that the unique solution of \autoref{eq:EL_stronganchoring}, with strong anchoring and $\epsilon$ sufficiently large, is well approximated by the unique solution of the limiting problem \autoref{eq:EL_limit_infty}.
\begin{prop}
  \label{prop:1}
  Let $\tilde{\mathbf{Q}} _{\epsilon} \in W ^{1, 2} \( \tilde{\Omega}; S _{0} \)$ be a solution of the LdG Euler-Lagrange equations \autoref{eq:EL_stronganchoring} on the rectangular domain $\tilde{\Omega}$ for $\epsilon > 0$, subject to the boundary condition $\( \( \tilde{Q} _{\epsilon} \) _{11}, \( \tilde{Q} _{\epsilon} \) _{12} \) = \tilde{\mathbf{g}} _{d}$ on $\partial \tilde{\Omega}$ for a given $0 < d < \min \lbk a / 4, b / 4 \rbk$.
  Then $\tilde{\mathbf{Q}} _{\epsilon}$ converges to the unique solution $\tilde{\mathbf{Q}} _{\infty}$ of \autoref{eq:EL_limit_infty} as $\epsilon \to \infty$, subject to the same Dirichlet condition with error estimates
  \begin{align} 
    \forall i = 1, 2, \quad \lnm \( \tilde{Q} _{\epsilon} \) _{1i} - \( \tilde{Q} _{\infty} \) _{1i} \rnm _{L ^{\infty} \( \tilde{\Omega} \)} \leq C \epsilon ^{-2},
  \end{align}
  for a positive constant $C$ independent of $\epsilon$.
\end{prop}

\begin{proof}[Proof of \autoref{prop:1}]
  By \cite{majumdar2010landau} (proposition 13), we have that $\tilde{\mathbf{Q}} _{\epsilon} \in C ^{\infty} \( \tilde{\Omega}; S _{0} \)$.
  Similarly, we have that $\tilde{\mathbf{Q}} _{\infty} \in C ^{\infty} \( \tilde{\Omega}; S _{0} \)$.
  Since $\lnm \( \tilde{g} _{d} \) _{i} \rnm _{L ^{\infty} \( \partial \tilde{\Omega} \)} \leq 1$ ($i = 1, 2$), 
  we have $\lnm \( \tilde{Q} _{\epsilon} \) _{1i} \rnm _{L ^{\infty} \( \tilde{\Omega} \)} \leq 1$ ($i = 1, 2$) \cite{bethuel1993asymptotics} (proposition 2).
  Therefore, comparing \autoref{eq:EL_stronganchoring} and \autoref{eq:EL_limit_infty}, we have,
  \begin{align}
    \forall i = 1, 2, & \quad - \epsilon ^{-2} \leq \Delta \( \( \tilde{Q} _{\epsilon} \) _{1i} - \( \tilde{Q} _{\infty} \) _{1i} \) \leq \epsilon ^{-2} , & \quad \tilde{\Omega}, \nonumber \\
    \forall i = 1, 2, & \quad \( \tilde{Q} _{\epsilon}  \) _{1i} - \( \tilde{Q} _{\infty} \) _{1i} = 0, & \quad \partial \tilde{\Omega}.
  \end{align}
  Let $v \in C ^{\infty} \( \tilde{\Omega}; \mathbb{R} \)$ be a solution of,
  \begin{align}
    \Delta v = 1 , & \quad \tilde{\Omega},  \nonumber \\
    v = 0, & \quad \partial \tilde{\Omega}.
  \end{align}
  The solution $v$ only depends the domain $\tilde{\Omega}$ which is fixed. We can check that $-\epsilon ^{-2} v$ and $\epsilon ^{-2} v$ are the sub-solution and the super-solution for both matrix components of $\( \tilde{\mathbf{Q}} _{\epsilon} - \tilde{\mathbf{Q}} _{\infty} \)$ and the result then follows, i.e.,
  \begin{align}
    \forall i = 1, 2, \quad \lnm \( \tilde{Q} _{\epsilon} \) _{1i} - \( \tilde{Q} _{\infty} \) _{1i} \rnm _{L ^{\infty} \( \tilde{\Omega} \)} \sim \bigO{\epsilon ^{-2}}.
  \end{align}
\end{proof}

We can compute the solution of the limiting problem \autoref{eq:EL_limit_infty} exactly for a given Dirichlet boundary condition as shown below.

\begin{prop}
  \label{prop:2}
  Let $\tilde{\mathbf{Q}} _{\infty}$ be the unique solution of \autoref{eq:EL_limit_infty}, on the rectangular domain $\tilde{\Omega}$ defined above, subject to the boundary condition $\( \( \tilde{Q} _{\infty} \) _{11}, \( \tilde{Q} _{\infty} \) _{12} \) = \tilde{\mathbf{g}} _{d}$, on $\partial \tilde{\Omega}$ for given $0 < d < \min \lbk a / 4, b / 4 \rbk$.
  Then we have $\( \tilde{Q} _{\infty} \) _{12} \equiv 0$ in $\tilde{\Omega}$ and, for all $\( x, y \) \in \tilde{\Omega}$,
  \begin{align}
    & \quad \( \tilde{Q} _{\infty} \) _{11} \( x, y \) \nonumber \\
    & = \sum _{k ~ \mathrm{odd}} \frac{4 \sin \( k \pi d / a \)}{k ^{2} \pi ^{2} d / a} \sin \( \frac{k \pi x}{a} \) \frac{\sinh \( k \pi \( b - y \) / a \) + \sinh \( k \pi y / a \)}{\sinh \( k \pi b / a \)} \nonumber \\
    & - \sum _{k ~ \mathrm{odd}} \frac{4 \sin \( k \pi d / b \)}{k ^{2} \pi ^{2} d / b} \sin \( \frac{k \pi y}{b} \) \frac{\sinh \( k \pi \( a - x \) / b \) + \sinh \( k \pi x / b \)}{\sinh \( k \pi a / b \)}.
    \label{eq:prop2_Q11}
  \end{align}
\end{prop}

\begin{proof}[Proof of \autoref{prop:2}]
  The existence of a unique solution for the Laplace equation with Dirichlet conditions, as in \autoref{eq:EL_limit_infty} is standard \cite{evans2010partial}. Let $\tilde{\mathbf{Q}} _{\infty}$ be the corresponding unique solution of \autoref{eq:EL_limit_infty} with boundary condition $\( \( \tilde{Q} _{\infty} \) _{11}, \( \tilde{Q} _{\infty} \) _{12} \) = \tilde{\mathbf{g}} _{d}$ on $\partial \tilde{\Omega}$, where $\( \tilde{Q} _{\infty} \) _{12} \equiv 0$ on $\partial \tilde{\Omega}$ from our choice of $\tilde{\mathbf{g}} _{d}$ above.
  Hence, we immediately have $\( \tilde{Q} _{\infty} \) _{12} \equiv 0$ in $\interior{\tilde{\Omega}}$.
  To compute $\( \tilde{Q} _{\infty} \) _{11}$, we solve the following boundary-value problem for $x \in \( 0, a \)$ and $y \in \( 0, b \)$,
  \begin{align}
    \Delta \( \tilde{Q} _{\infty} \) _{11} \( x, y \) & = 0,  \nonumber\\
    \( \tilde{Q} _{\infty} \) _{11} \( x, 0 \) = \( \tilde{Q} _{\infty} \) _{11} \( x, b \) & = + T _{d / a} \( x / a \), \nonumber \\
    \( \tilde{Q} _{\infty} \) _{11} \( a, y \) = \( \tilde{Q} _{\infty} \) _{11} \( 0, y \) & = - T _{d / b} \( y / b \).
  \end{align}

  Consider the scalar function $f \( x, y; a, b \)$ which satisfies the following boundary-value problem,
  \begin{align}
    \Delta f \( x, y; a, b \) & = 0, \nonumber \\
    f \( x, 0; a, b \) & = T _{d / a} \( x / a \), \nonumber \\
    f \( a, y; a, b \) = f \( x, b; a, b \) = f \( 0, y; a, b \) & = 0.
  \end{align}

  We solve for $f \( x, y; a, b \)$ by separation of variables (also see \cite{lewis2015defects}). A standard computation shows, using boundary conditions at $x = 0$ and $x = a$,
  \begin{align}
    f \( x, y; a, b \) = \sum _{k = 1} ^{+ \infty} \sin \frac{k \pi x}{a} \( c _{k}' \( \rme ^{k \pi y / a} - \rme ^{-k \pi y / a} \) + c _{k}'' \( \rme ^{k \pi \( y - b \) / a} - \rme ^{-k \pi \( y - b \) / a} \) \).
  \end{align}
  where $c _{k}'$ and $c _{k}''$ are constants depend on $k$.
  We use the remaining boundary conditions at $y = 0$ and $y = b$ and Fourier series methods to obtain
  \begin{align}
    \quad f \( x, y; a, b \)
    & = \sum _{k ~ \mathrm{odd}} \frac{4 \sin \( k \pi d / a \)}{k ^{2} \pi ^{2} d / a} \sin \frac{k \pi x}{a} \frac{\rme ^{k \pi \( y - b \) / a} - \rme ^{-k \pi \( y - b \) / a}}{\rme ^{-k \pi b / a} - \rme ^{k \pi b / a}}, \nonumber \\
    & = \sum _{k ~ \mathrm{odd}} \frac{4 \sin \( k \pi d / a \)}{k ^{2} \pi ^{2} d / a} \sin \frac{k \pi x}{a} \frac{\sinh \( k \pi \( b - y \) / a \)}{\sinh \( k \pi b / a \)}.
    \label{eq:prop3_f}
  \end{align}
  Then, it is straightforward to verify that
  \begin{align}
    \( \tilde{Q} _{\infty} \) _{11} \( x, y \) = f \( x, y; a, b \) - f \( y, a - x; b, a \) + f \( x, b - y; a, b \) - f \( y, x; b, a \).
  \end{align}
  and the result \autoref{eq:prop2_Q11} now follows.
\end{proof}

On a square with $a = b$, the WORS solution is distinguished by $\tilde{\mathbf{Q}} = 0$ along the square diagonals, in particular $\tilde{\mathbf{Q}} \( \frac{a}{2}, \frac{a}{2} \) = 0$ at the square centre. It is reasonable to ask if $\tilde{\mathbf{Q}} _{\infty}$ retains these nodal diagonal lines for $a \neq b$. A simple test is to check if $\tilde{\mathbf{Q}} _{\infty} \( \frac{a}{2}, \frac{b}{2} \) \neq 0$ which would demonstrate the loss of the WORS structure for $a \neq b$. In \autoref{fig:WORS_sBD2}, we plot $\tilde{\mathbf{Q}} _{\infty}$ on a square ($a = b = 1$) which is simply the WORS and $\tilde{\mathbf{Q}} _{\infty}$ on a rectangle ($a = 1.5$, $b = 1$) to illustrate the differences to the reader.

\begin{figure}[ht]
  \centering
  \includegraphics[width = 0.5\textwidth] {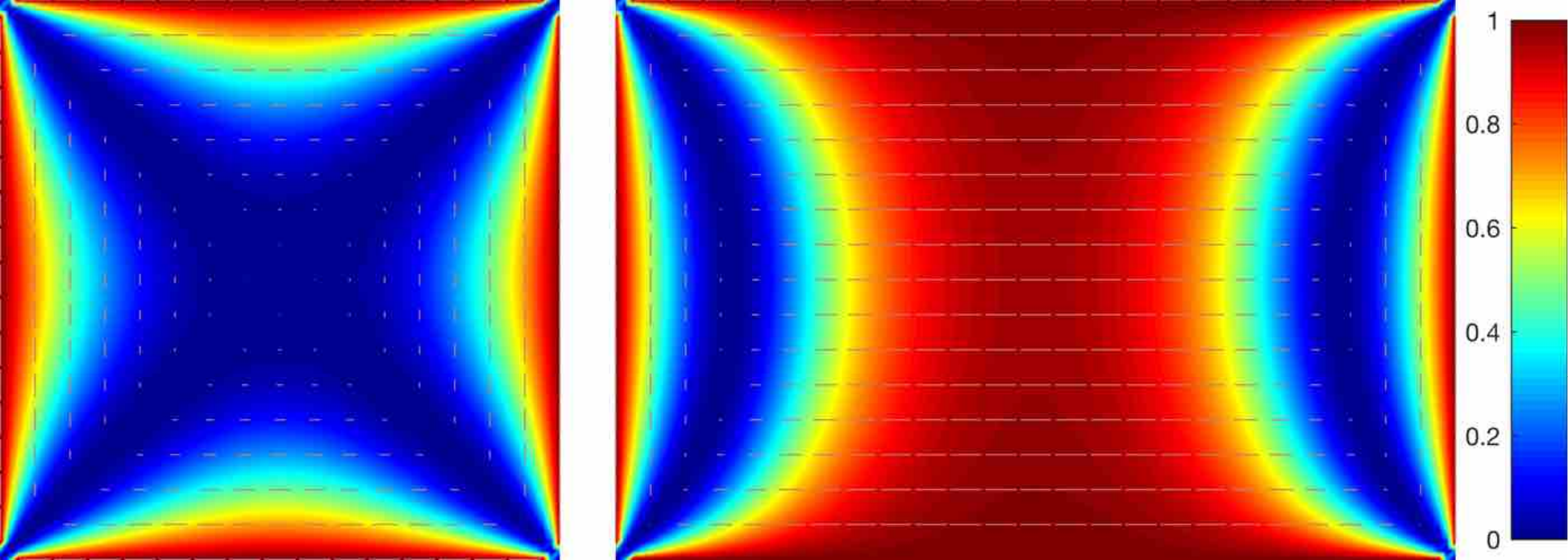} \\
  \caption{Left: WORS at domain size $1 \times 1$, which is marked as `B' in \autoref{fig:Bifurcation__LC_bifurcation_N4_h64_64_120945_E} and \autoref{fig:Bifurcation__LC_bifurcation_N4_h64_64_120945_Q}. Right: sBD2 at domain size $1.5 \times 1$, which is marked as `B' in \autoref{fig:Bifurcation__LC_bifurcation_N4_h96_64_120945_E} and \autoref{fig:Bifurcation__LC_bifurcation_N4_h96_64_120945_Q}. Mesh spacing is $h = 1 / 64$. We fix $\( \( \tilde{Q} _{\infty} \) _{11}, \( \tilde{Q} _{\infty} \) _{12} \) = \tilde{\mathbf{g}} _{d}$ on $\partial \Omega$ with $d = 0.03$. The colour bar represents the value of $\tilde{s} _{\infty} ^{2} = \tr{\tilde{\mathbf{Q}} _{\infty} ^{2}} / 2$.}
  \label{fig:WORS_sBD2}
\end{figure}

\begin{prop}
  \label{prop:3}
  Let $\tilde{\mathbf{Q}} _{\infty}$ be the unique solution of \autoref{eq:EL_limit_infty} on the rectangular domain $\tilde{\Omega}$, subject to the boundary condition $\( \( \tilde{Q} _{\infty} \) _{11}, \( \tilde{Q} _{\infty} \) _{12} \) = \tilde{\mathbf{g}} _{d}$, for a given $0 < d < \min \lbk a / 4, b / 4 \rbk$.
  For $a > b = 1$, we have
  \begin{align}
    \( \tilde{Q} _{\infty} \) _{11} \( \frac{a}{2}, \frac{b}{2} \) > 0,
  \end{align}
  so that we lose the cross structure of the WORS on all rectangles (similar arguments also apply to $0 < a < b = 1$).
\end{prop}

\begin{proof}[Proof of \autoref{prop:3}]
  Let $\tilde{\mathbf{Q}} _{\infty} '$ denote the WORS solution of \autoref{eq:EL_limit_infty} on a square domain $\Omega _{b, b}$ with boundary condition $\( \( \tilde{Q} _{\infty} ' \) _{11}, \( \tilde{Q} _{\infty} ' \) _{12} \) = \tilde{\mathbf{g}} _{d; b, b}$ on $\partial \Omega _{b, b}$.
  To prove the result above, we show that $\( \tilde{Q} _{\infty} \) _{11} \( a / 2, b / 2 \) > \( \tilde{Q} _{\infty} ' \) _{11} \( b / 2, b / 2 \) = 0$. 
  We construct five auxiliary boundary value problems below and we illustrate the boundary conditions in \autoref{fig:rectangularedge_12345}, using the boundary function in \autoref{eq:boundary_T}. Let $I _{p} = \lsb - p / 2, p / 2 \rsb$ for $p > 0$.
  \begin{align}
    &
    \left\{
    \begin{array}{ll}
      \Delta u _{1} \( x, y \) = 0, & \quad \( x, y \) \in \Omega _{1} := I _{a} \times I _{b}, \\
      u _{1} \( x, y \) = +T _{d / a} \( x / a + 1 / 2 \), & \quad x \in I _{a}, y \in \lbk - b / 2, b / 2 \rbk, \\
      u _{1} \( x, y \) = 0, & \quad x \in \lbk - a / 2, a / 2 \rbk, y \in I _{b}.
    \end{array}
    \right.
    \\
    &
    \left\{
    \begin{array}{ll}
      \Delta u _{2} \( x, y \) = 0, & \quad \( x, y \) \in \Omega _{2} := I _{b} \times I _{b}, \\
      u _{2} \( x, y \) = +T _{d / b} \( x / b + 1 / 2 \), & \quad x \in I _{b}, y \in \lbk - b / 2, b / 2 \rbk, \\
      u _{2} \( x, y \) = 0, & \quad x \in \lbk - b / 2, b / 2 \rbk, y \in I _{b}.
    \end{array}
    \right.
    \\
    &
    \left\{
    \begin{array}{ll}
      \Delta u _{3} \( x, y \) = 0, & \quad \( x, y \) \in \Omega _{3} := I _{a} \times I _{b}, \\
      u _{3} \( x, y \) = 0, & \quad x \in I _{a}, y \in \lbk - b / 2, b / 2 \rbk, \\
      u _{3} \( x, y \) = -T _{d / b} \( y / b + 1 / 2 \), & \quad x \in \lbk - a / 2, a / 2 \rbk, y \in I _{b}.
    \end{array}
    \right.
    \\
    &
    \left\{
    \begin{array}{ll}
      \Delta u _{4} \( x, y \) = 0, & \quad \( x, y \) \in \Omega _{4} := I _{a} \times I _{a}, \\
      u _{4} \( x, y \) = 0, & \quad x \in I _{a}, y \in \lbk - a / 2, a / 2 \rbk, \\
      u _{4} \( x, y \) = -T _{d / a} \( y / a + 1 / 2 \), & \quad x \in \lbk - a / 2, a / 2 \rbk, y \in I _{a}.
    \end{array}
    \right.
    \\
    &
    \left\{
    \begin{array}{ll}
      \Delta u _{5} \( x, y \) = 0, & \quad \( x, y \) \in \Omega _{5} := I _{b} \times I _{b}, \\
      u _{5} \( x, y \) = 0, & \quad x \in I _{b}, y \in \lbk - b / 2, b / 2 \rbk, \\
      u _{5} \( x, y \) = -T _{d / b} \( y / b + 1 / 2 \), & \quad x \in \lbk - b / 2, b / 2 \rbk, y \in I _{b}.
    \end{array}
    \right.
  \end{align}

  Using the maximum principle \cite{evans2010partial}, we have $u _{1} \geq 0$ in $\Omega _{1}$ and $u _{2} \geq 0$ in $\Omega _{2}$, so that, $u _{1} > u _{2}$ on the vertical lines $x = \pm b / 2$. Hence,
  \begin{align}
    u _{1} \( x, y \) \geq u _{2} \( x, y \) & \quad \( x, y \) \in \partial \Omega _{2}.
  \end{align}
  By the strong maximum principle, $u _{1} \( 0, 0 \) > u _{2} \( 0, 0 \)$.

  Similarly, we have that $u _{3} \leq 0$ on $\Omega _{3}$ and $u _{4} \leq 0$ in $\Omega _{4}$, so that
  \begin{align}
    u _{3} \( x, y \) \geq u _{4} \( x, y \) & \quad \( x, y \) \in \partial \Omega _{3}.
  \end{align}
  By the strong maximum principle, $u _{3} \( 0, 0 \) > u _{4} \( 0, 0 \)$. 
  We can also check that $u _{4} \( x, y \) = u _{5} \( xb / a, yb / a \)$, for $\( x, y \) \in \Omega _{4}$, since the Laplace operator is invariant with respect to uniform scaling and hence, $u _{4} \( 0, 0 \) = u _{5} \( 0, 0 \)$.

  By superposition of boundary conditions, we see that the functions $\( \tilde{Q} _{\infty} \) _{11}$ and $u _{6} := u _{1} + u _{3}$ only differ by a translation,
  \begin{align}
    \( \tilde{Q} _{\infty} \) _{11}   \( x + \frac{a}{2}, y + \frac{b}{2} \) = u _{6} \( x, y \), & \quad \( x, y \) \in \Omega _{1} = \Omega _{3},
  \end{align}
  and we have the similar relationship for $\( \tilde{Q} _{\infty} \) _{11} '$ and $u _{7} := u _{2} + u _{5}$,
  \begin{align}
    \( \tilde{Q} _{\infty} \) _{11} ' \( x + \frac{b}{2}, y + \frac{b}{2} \) = u _{7} \( x, y \), & \quad \( x, y \) \in \Omega _{2} = \Omega _{5}.
  \end{align}
  The result immediately follows since
  \begin{align}
    \( \tilde{Q} _{\infty} \) _{11} \( \frac{a}{2}, \frac{b}{2} \) = u _{6} \( 0, 0 \) > u _{7} \( 0, 0 \) = \( \tilde{Q} _{\infty} \) _{11} ' \( \frac{b}{2}, \frac{b}{2} \) = 0.
  \end{align}
\end{proof}

\begin{figure}[ht]
  \centering
  \includegraphics[width = 0.5\textwidth] {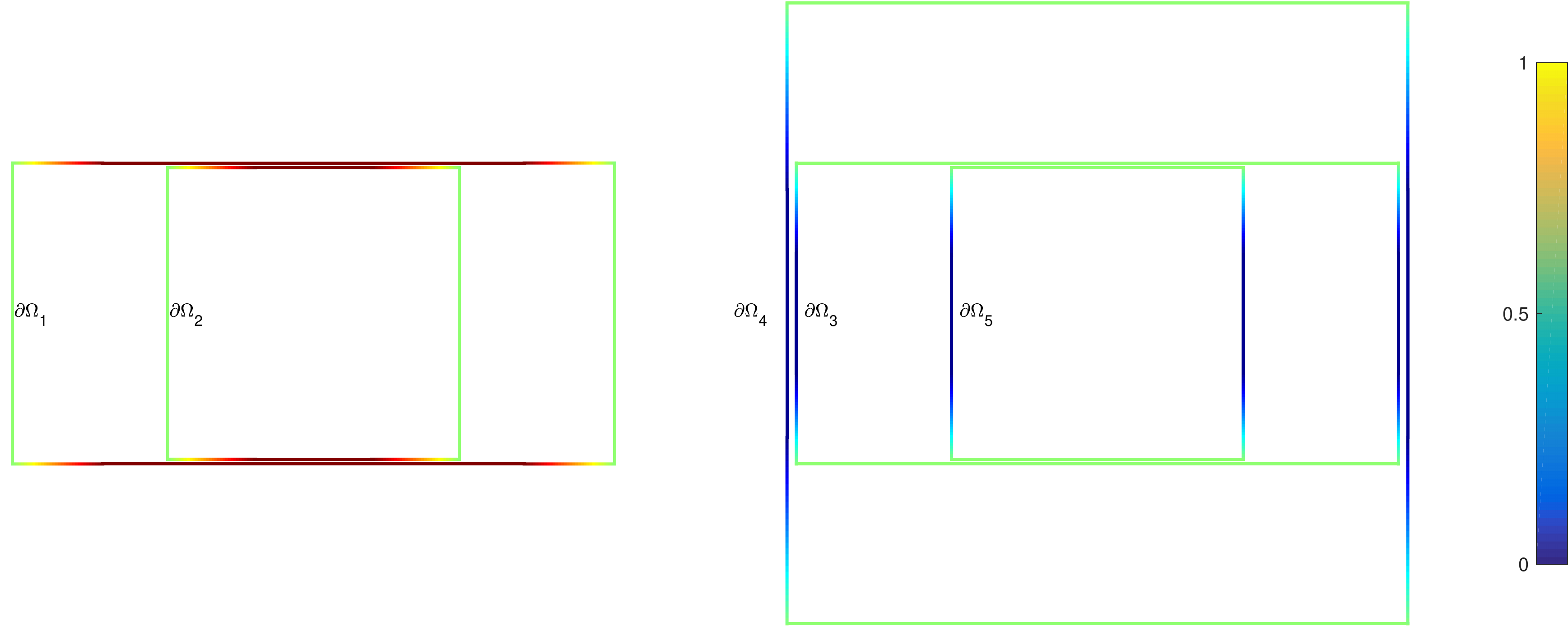} \\
  \caption{Dirichlet boundary conditions for rectangular domains. The domain parameter is $a = 2b$ and the boundary parameter is $d = 0.3$ here. Left: boundary values for $u _{1}$ on $\partial \Omega _{1}$ (rectangular boundary) and $u _{2}$ on $\partial \Omega _{2}$ (square boundary) are represented by different colours. Right: boundary values for $u _{3}$ on $\partial \Omega _{3}$ (rectangular boundary), $u _{4}$ on $\partial \Omega _{4}$ (large square boundary) and $u _{5}$ on $\partial \Omega _{5}$ (small square boundary) are represented by different colours. The domains are scaled a little bit for a better illustration.}
  \label{fig:rectangularedge_12345}
\end{figure}

\begin{rmkplain}
  In the limit $d \to 0$, the function $\( \tilde{Q} _{\infty} \) _{11}$ converges pointwise to
  \begin{align}
    \( \tilde{Q} _{\infty} \) _{11} \( x, y \)
    & = \sum _{k ~ \mathrm{odd}} \frac{4}{k \pi} \sin \( \frac{k \pi x}{a} \) \frac{\sinh \( k \pi \( b - y \) / a \) + \sinh \( k \pi y / a \)}{\sinh \( k \pi b / a \)} \nonumber \\
    & - \sum _{k ~ \mathrm{odd}} \frac{4}{k \pi} \sin \( \frac{k \pi y}{b} \) \frac{\sinh \( k \pi \( a - x \) / b \) + \sinh \( k \pi x / b \)}{\sinh \( k \pi a / b \)},
  \end{align}
  which is still smooth in $\interior{\Omega}$, so that the maximum principle is valid. Thus, the result in \autoref{prop:3}  holds in the $d \to 0$ limit for $a > b$.
\end{rmkplain}

\subsubsection{Weak Anchoring Condition}

In \autoref{prop:3}, we show the  perfect cross structure of the WORS on a square domain is lost as soon as we introduce geometrical anisotropy with $a \neq b$, at least with Dirichlet boundary conditions or strong anchoring. It is equally natural to ask if the WORS is an artefact of the strong anchoring condition  i.e. do we lose the cross structure as soon as we relax the boundary conditions on the edges of the square or do we have a continuous transition from the WORS to a cross-free profile, as the anchoring is relaxed by means of a suitably defined anchoring parameter? The `cross' structure refers to the fact that the WORS vanishes along the square diagonals and we refer to the nodal diagonal lines as a `cross'.

We work with $\mathbf{Q}$-tensors (defined in \autoref{eq:Q_3D_2D}) that are independent of the vertical dimension, on a square well $\Omega$ with $a = b = 1$ and $L > 0$. The key difference is that we impose a Durand-Nobili surface energy \cite{nobili1992disorientation} on the lateral surfaces $\partial \Omega \times \lsb 0, h \rsb$ (instead of a Dirichlet uniaxial tangent condition) with 
\begin{align}
  f _{\textrm{S}} \( \mathbf{Q} \) = W \lmdl \( Q _{11} , Q _{12} \) - \frac{B}{2C} \mathbf{g} _{d} \rmdl ^{2},
\end{align}
for an anchoring coefficient $W > 0$, and $\mathbf{g} _{d} \( x, y \) = \tilde{\mathbf{g}} _{d} \( x / L, y / L \)$, $\( x, y \) \in \partial \Omega$. In the limit $W \to \infty$, we recover the strong anchoring condition in \autoref{eq:boundary_g}. 
Using the standard scalings, the rescaled LdG free energy is
\begin{align}
  & \frac{1}{h K} F \lsb \bar{\mathbf{Q}} \rsb \nonumber \\
  & := \int _{\tilde{\Omega}} \( \frac{1}{2} \lmdl \nabla \bar{\mathbf{Q}} \rmdl ^{2} + \frac{L ^{2}}{K} \( -\frac{B ^{2}}{4C} \tr{\bar{\mathbf{Q}} ^{2}} + \frac{C}{4} \( \tr{\bar{\mathbf{Q}} ^{2}} \) ^{2} \) - \frac{5}{432} \frac{B ^{4} L ^{2}}{C ^{3} K} \)
  + \int_{\partial \tilde{\Omega}} \frac{WL}{K} \lmdl \( \bar{Q} _{11} , \bar{Q} _{12} \) - \frac{B}{2C} \tilde{\mathbf{g}} _{d} \rmdl ^{2}.
\end{align}
Defining $\tilde{\mathbf{Q}}$ as in \autoref{eq:Qtilde} and $\epsilon$ as in \autoref{eq:epsilon}, we work with a reduced and rescaled LdG free energy of the form
\begin{align}
  \frac{4 C ^{2}}{B ^{2} h K} \tilde{F} \lsb \tilde{\mathbf{Q}} \rsb 
  & := \int_{\tilde{\Omega}} \( \lmdl \nabla \tilde{Q} _{11} \rmdl ^{2} + \lmdl \nabla \tilde{Q} _{12} \rmdl ^{2} + \frac{1}{\epsilon ^{2}} \( \tilde{Q} _{11} ^{2} + \tilde{Q} _{12} ^{2} - 1 \) ^{2} - \frac{32}{27} \frac{1}{\epsilon ^{2}} \) 
  + \tilde{W} \int _{\partial \tilde{\Omega}}  \lmdl \( \tilde{Q} _{11}, \tilde{Q} _{12} \) - \tilde{\mathbf{g}} _{d} \rmdl ^{2},
  \label{eq:LdG_rescaled_weak}
\end{align}
where the reduced anchoring coefficient is
\begin{align}
  \tilde{W} = \frac{WL}{K} = \frac{W}{\epsilon} \sqrt{\frac{4 C}{B ^{2} K}}.
\end{align}
We treat $B$, $C$, $K$ to be fixed material constants in this manuscript and hence, we have $\tilde{W} \propto W \epsilon ^{-1}$. Our aim is to investigate the effects of a mild relaxation of the strong anchoring conditions and hence, we assume that $W = \bigO{\epsilon}$ as $\epsilon \to \infty$ so that $\tilde{W} = \alpha$ for some positive constant $\alpha$ independent of $\epsilon$, for large values of $\epsilon$. 
We can readily prove existence of minimizers of \autoref{eq:LdG_rescaled_weak} under these assumptions from the direct methods in the calculus of variations and for sufficiently large $\epsilon$, we can follow the arguments in \cite{lamy2014bifurcation} and \cite{canevari2017order} to demonstrate the uniqueness of critical points of \autoref{eq:LdG_rescaled_weak} in this limit.
It is straightforward to show that minimizers of \autoref{eq:LdG_rescaled_weak} satisfy the following system of partial differential equations:
\begin{align}
  \label{eq:LdG_rescaled_weak_EL1}
  -\Delta \( \tilde{Q} _{11}, \tilde{Q} _{12} \) + \frac{1}{\epsilon ^{2}} \( \lmdl \( \tilde{Q} _{11}, \tilde{Q} _{12} \) \rmdl ^{2} - 1 \) \( \tilde{Q} _{11}, \tilde{Q} _{12} \) & = 0, & \tilde{\Omega}, \\
  \label{eq:LdG_rescaled_weak_EL2}
  \mathbf{\nu} _{\tilde{\Omega}}' \nabla \( \tilde{Q} _{11}, \tilde{Q} _{12} \) + \alpha \( \( \tilde{Q} _{11}, \tilde{Q} _{12} \) - \tilde{\mathbf{g}} _{d} \) & = 0, & \partial \tilde{\Omega},
\end{align} 
where $\mathbf{\nu} _{\tilde{\Omega}}$ is the normal to the edges of the re-scaled rectangular domain, $\tilde{\Omega}$. 
Following the methods in \cite{canevari2017order}, we can prove the existence of a Well Order Reconstruction Solution (WORS) for the system \autoref{eq:LdG_rescaled_weak_EL1} and \autoref{eq:LdG_rescaled_weak_EL2} with $\tilde{Q} _{12} = 0$ everywhere and $\tilde{Q} _{11} = 0$ on the square diagonals and the cross-structure is preserved.

Under these assumptions, we can also check that minimizers of \autoref{eq:LdG_rescaled_weak} converge uniformly to minimizers of the following limiting energy as $\epsilon \to \infty$:
\begin{align}
  \frac{4 C ^{2}}{B ^{2} h K} \tilde{F} _{\infty} \lsb \tilde{\mathbf{Q}} _{\infty} \rsb 
  & := \int_{\tilde{\Omega}} \( \lmdl \nabla \( \tilde{Q} _{\infty} \) _{11} \rmdl ^{2} + \lmdl \nabla \( \tilde{Q} _{\infty} \) _{12} \rmdl ^{2} \) + \alpha \int _{\partial \tilde{\Omega}}  \lmdl \( \( \tilde{Q} _{\infty} \) _{11}, \( \tilde{Q} _{\infty} \) _{12} \) - \tilde{\mathbf{g}} _{d} \rmdl ^{2}.
  \label{eq:LdG_rescaled_weak_limit_infinity}
\end{align} 
We can compute the limiting profile, $\tilde{\mathbf{Q}} _{\infty}$, exactly in terms of its matrix components $\( \tilde{Q} _{\infty} \) _{11}$ and $\( \tilde{Q} _{\infty} \) _{12}$, as shown below.
We define the following function $J$ for $q \in W ^{1, 2} \( \tilde{\Omega} \) \cap C ^{\infty} \( \interior{\tilde{\Omega}} \)$, $\tau \in \mathbb{R}$, $g \in C ^{\infty} \( \partial \tilde{\Omega} \)$,
\begin{align}
  J \( q; \tau, g \) := \int _{\tilde{\Omega}} \lmdl \nabla q \rmdl ^{2} + \int _{\partial \tilde{\Omega}} \tau \( q - g \) ^{2},
\end{align}
If $q$ is a critical point (e.g., local minimizer) of the energy $J$, then we necessarily have,
\begin{align}
  \label{eq:weakanchoring_q}
  \left\{
  \begin{array}{rr}
    \Delta q = 0, & \tilde{\Omega}, \\
    \mathbf{\nu} _{\tilde{\Omega}} \cdot \nabla q + \tau \( q - g \) = 0, & \partial \tilde{\Omega}.
  \end{array}
  \right.
\end{align}
This is simply the Laplace equation on a rectangular domain with Robin boundary conditions, which can be solved by means of separation of variables.
It follows that
\begin{align}
  \( \tilde{Q} _{\infty} \) _{11} & = \arg \min _{q} J \( q; \tilde{W}, \( \tilde{g} _{d} \) _{1} \), 
  \label{eq:weakanchoring_EL_Q11} \\
  \( \tilde{Q} _{\infty} \) _{12} & = \arg \min _{q} J \( q; \tilde{W}, \( \tilde{g} _{d} \) _{2} \). 
  \label{eq:weakanchoring_EL_Q12}
\end{align} 
For $\tau > 0$, the boundary-value problem \autoref{eq:weakanchoring_q} has a unique solution which can be computed exactly. 

\begin{rmkplain}[Cross structure of $\tilde{\mathbf{Q}} _{\infty}$ with $\tilde{W} > 0$ for a square with $a = b = 1$]
  We have a unique solution $\( \tilde{Q} _{\infty} \) _{11}$ for \autoref{eq:weakanchoring_EL_Q11} and $\( \tilde{Q} _{\infty} \) _{12}$ for \autoref{eq:weakanchoring_EL_Q12} for $\tau > 0$. Since $\( \tilde{g} _{d} \) _{2} \equiv 0$, this yields $\( \tilde{Q} _{\infty} \) _{12} \equiv 0$ on $\tilde{\Omega}$. Using the antisymmetry of $\( \tilde{g} _{d} \) _{1}$, we have
  \begin{align}
    \( \tilde{g} _{d} \) _{1} \( x, y \) = -\( \tilde{g} _{d} \) _{1} \( y, x \) = -\( \tilde{g} _{d} \) _{1} \( 1 - y, 1 - x \), \quad \( x, y \) \in \partial \tilde{\Omega}
  \end{align}
  which implies that
  \begin{align}
    \( \tilde{Q} _{\infty} \) _{11} \( x, y \) = -\( \tilde{Q} _{\infty} \) _{11} \( y, x \) = -\( \tilde{Q} _{\infty} \) _{11} \( 1 - y, 1 - x \), \quad \( x, y \) \in \tilde{\Omega}.
  \end{align}
  Hence, we obtain
  \begin{align}
    \( \tilde{Q} _{\infty} \) _{11} \( x, x \) = 0 = \( \tilde{Q} _{\infty} \) _{11} \( x, 1 - x \), \quad \forall x \in \( 0, 1 \)
  \end{align} 
  from these relations quite easily. 
  This simple argument demonstrates that $\tilde{\mathbf{Q}} _{\infty}$ vanishes along the two square diagonals $y = x$ and $y = 1 - x$ and the cross structure is retained in the limit too, under the assumption that $\tilde{W}$ remains bounded as $\epsilon \to \infty$.
\end{rmkplain}

\begin{rmkplain}[Explicit limiting solution with weak anchoring]
  Next, we analytically solve the limiting problem \autoref{eq:weakanchoring_q} for $\( \tilde{Q} _{\infty} \) _{11}$. 
  In the  $d \to 0$ limit, the boundary function $\( g _{d} \) _{1}$ is either $+1$ or $-1$ on all four edges. Using ideas similar to \autoref{prop:2}, we can write $\( \tilde{Q} _{\infty} \) _{11}$ as
  \begin{align}
    \( \tilde{Q} _{\infty} \) _{11} \( x, y \) = f \( x, b - y; a, b \) - f \( y, x; b, a \) + f \( x, y; a, b \) - f \( y, a - x; b, a \).
    \end{align}
  where the function $f \in W ^{1, 2} \( \tilde{\Omega} \) \cap C ^{\infty} \( \interior{\tilde{\Omega}} \)$ is a solution of the following boundary-value problem,
  \begin{align}
    \Delta f \( x, y; a, b \) = 0, & \quad \( x, y \) \in \tilde{\Omega}, \\
    \tau f \( x, y; a, b \) - \frac{\partial f \( x, y; a, b \)}{\partial y} = 0, & \quad y = 0, \\
    \tau f \( x, y; a, b \) + \frac{\partial f \( x, y; a, b \)}{\partial y} = \tau, & \quad y = b, \\
    \tau f \( x, y; a, b \) - \frac{\partial f \( x, y; a, b \)}{\partial x} = 0, & \quad x = 0, \\
    \tau f \( x, y; a, b \) + \frac{\partial f \( x, y; a, b \)}{\partial x} = 0, & \quad x = a.
  \end{align}
  We can solve for $f$ by separation of variables. Skipping all the technical details, we have
  \begin{align}
    & \quad f \( x, y; a, b \) = \sum _{k = 1} ^{\infty} E _{k} X _{k} \( x \) Y _{k} \( y \), \\
    & = \sum _{k = 1} ^{\infty} \( \frac{2}{p _{k} ^{2} a + \tau ^{2} a + 2 \tau} \) \( p _{k} \cos \( p _{k} x \) + \tau \sin \( p _{k} x \) \) \nonumber \\
    & \quad \tau \frac{\cos \( p _{k} a \) \( p _{k} ^{2} + \tau ^{2} \) + \( p _{k} ^{2} - \tau ^{2} \)}{p _{k} \( p _{k} ^{2} - \tau ^{2} \)} \tau \frac{p _{k} \cosh \( p _{k} y \) + \tau \sinh \( p _{k} y \)}{\( p _{k} ^{2} + \tau ^{2} \) \sinh \( p _{k} b \) + 2 \tau p _{k} \cosh \( p _{k} b \)}.
    \label{eq:weakanchoring_f}
  \end{align}
  where $\lbk p _{k} \rbk$ is the set of solutions of the transcendental equation,
  \begin{align}
    \tan \( pa \) = \frac{2 \tau p}{p ^{2} - \tau ^{2}}.
  \end{align}
  (Also see \cite{walton2018nematic} for similar calculations.)
\end{rmkplain}


To summarise, we investigate the effects of both geometrical anisotropy ($a \neq b$) and relaxed but strong anchoring on the WORS in this section. We lose the cross structure of the WORS for $a \neq b$. However, for $W = \bigO{\epsilon}$ and for $\epsilon$ sufficiently large, we retain the cross structure of the WORS so that the WORS is not an artefact of the Dirichlet conditions. These properties are preserved in the limit since we can compute the limiting profiles, $\tilde{\mathbf{Q}}_\infty$ in both cases.  Of course, as $W$ decreases and $\tilde{W} \to 0$ for $W \ll \epsilon$, then $\tilde{\mathbf{Q}} _{\infty}$ can be any constant matrix as expected, since there are no boundary constraints in this limit. In \autoref{fig:Weak_Anchoring__LC_weakanchoring_N4_h64_64_t3t10}, we plot $\tilde{\mathbf{Q}} _{\infty}$ on a square with $\tilde{W} = 3$ and $\tilde{W} = 10$ respectively. The cross structure is necessarily more diffuse with smaller values of $\tilde{W}$. In \autoref{fig:Weak_Anchoring__LC_weakanchoring_N4_h96_64_t3t10}, we illustrate similar numerical results for a rectangle with $a = 1.5$, $b = 1$; the cross structure disappears as we see nodal lines of $\tilde{\mathbf{Q}} _{\infty}$ along the shorter rectangular edges. We refer to such solutions as BD solutions in the rest of the text, by analogy with similar terminology in \cite{wang2019order}.

\begin{figure}[ht]
  \centering
  \includegraphics[width = 0.5\textwidth] {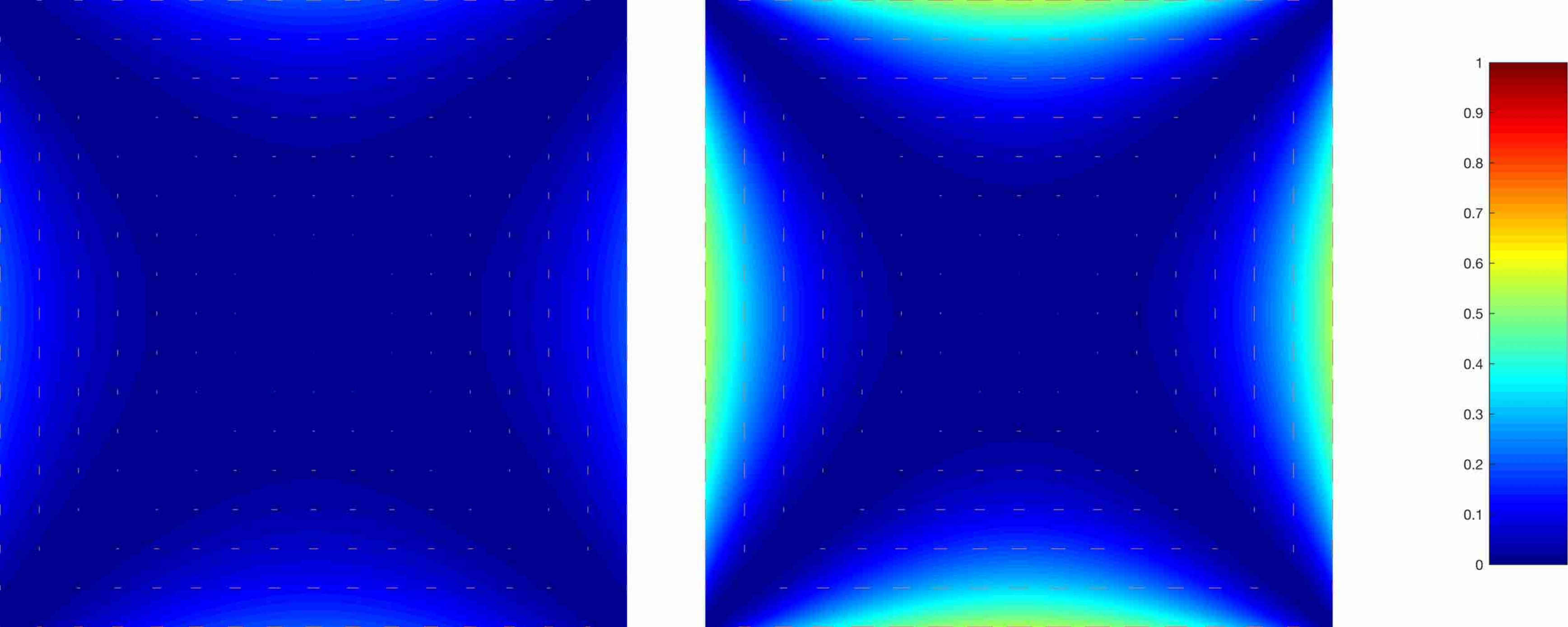} \\
  \caption{Solution of \autoref{eq:weakanchoring_q} for domain size $1 \times 1$, mesh spacing $h = 1 / 64$. Left: $\tilde{W} = 3$. Right: $\tilde{W} = 10$. The colour bar represents the value of $\tilde{s} _{\epsilon} ^{2} = \tr{\tilde{\mathbf{Q}} _{\epsilon} ^{2}} / 2$.}
  \label{fig:Weak_Anchoring__LC_weakanchoring_N4_h64_64_t3t10}
\end{figure}
\begin{figure}[ht]
  \centering
  \includegraphics[width = 0.75\textwidth] {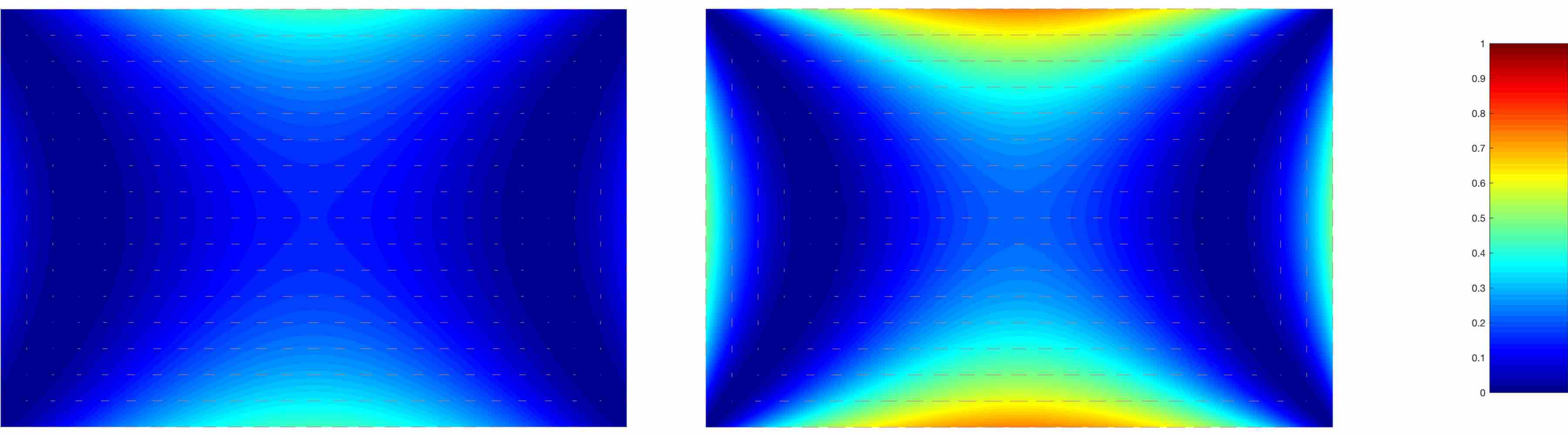} \\
  \caption{Solution of \autoref{eq:weakanchoring_q} for domain size $1.5 \times 1$, mesh spacing $h = 1 / 64$. Left: $\tilde{W} = 3$. Right: $\tilde{W} = 10$. The colour bar represents the value of $\tilde{s} _{\epsilon} ^{2} = \tr{\tilde{\mathbf{Q}} _{\epsilon} ^{2}} / 2$.}
  \label{fig:Weak_Anchoring__LC_weakanchoring_N4_h96_64_t3t10}
\end{figure}

\subsection{The \texorpdfstring{$\epsilon \to 0$}{epsilon to zero} Limiting Problem}
\label{sec:limits_OF}

Let $\tilde{\mathbf{Q}} _{\epsilon}$ be a global minimizer of the rescaled LdG free energy (\autoref{eq:LdG_2D_dimless}) subject to the fixed strong anchoring boundary conditions $\( \( \tilde{Q} _{\epsilon} \) _{11}, \( \tilde{Q} _{\epsilon} \) _{12} \) = \tilde{\mathbf{g}} _{d}$ on $\partial \tilde{\Omega}$ for given $d > 0$. For $\epsilon$ sufficiently small, in the interior, we have to leading order (more precisely, under suitable conditions, the error is $\bigO{\epsilon ^{2}}$ in $L ^{\infty}$ norm \cite{bethuel1993asymptotics}),
\begin{align}
  \tilde{\mathbf{Q}} _{\epsilon} \sim \( \uvec{n} \otimes \uvec{n} - \frac{\mathbf{I} _{2}}{2} \),
\end{align}
where $\uvec{n} = \( \cos \theta, \sin \theta \) \tp$ and $\theta$ is the solution of the Laplace equation
\begin{align}
  \label{eq:theta_EL}
  \Delta \theta \( x, y \) & = 0, \quad \( x, y \) \in \tilde{\Omega}.
\end{align}
subject to appropriately defined Dirichlet conditions \cite{lewis2015defects}, \cite{majumdar2010landau}. 

Hence, we can obtain good approximations to locally stable minimizers of \autoref{eq:LdG_2D_dimless} for sufficiently small $\epsilon$, by studying Dirichlet boundary-value problems for the angle $\theta$ on a rectangle as given below (we set $d = 0$ in the Dirichlet condition $\tilde{\mathbf{g}} _{d}$ for the sake of this computation), 
\begin{align}
  \label{eq:theta_EL_bdry}
  \theta \( x, 0 \) & = d _{1}, \quad \theta \( a, y \) = d _{2}, \quad \theta \( x, b \) = d _{3}, \quad \theta \( 0, y \) = d _{4},
\end{align}
where $d _{1}, d _{2}, d _{3}, d _{4}$ are arbitrary constants accounting for the constant boundary conditions on the edges of the rectangle. This is analogous to the approach in \cite{lewis2014colloidal}. 
In the case of a rectangle with $a > b$, it is known that there are three competing equilibria---the diagonal states for which $\uvec{n}$ is aligned along a diagonal of the rectangle, the rotated R1 and R2 states for which $\uvec{n}$ rotates by $\pi$ radians between a pair of parallel horizontal edges and the rotated R3, R4 states for which $\uvec{n}$ rotates by $\pi$ radians between a pair of parallel vertical edges. For $a > b$, the R3, R4 states have lower energies than the R1, R2 states (see \cite{tsakonas2007multistable}, \cite{lewis2014colloidal} for details). We can obtain good approximations to the diagonal and rotated solutions by studying solutions of \autoref{eq:theta_EL} subject to the boundary conditions enumerated in \autoref{tab:OF_boundarycondition}, which are consistent with \cite{luo2012multistability}. These boundary conditions enforce $\pm 1 / 4$ topological defects in the vector field, $\uvec{n}$, at the corners of $\tilde{\Omega}$. To be more precise, we say that a vertex defect has charge $t / 4$ ($t$ is an odd integer) when the director rotates by $2 \pi t / 4$ radians between a pair of adjacent edges; the sign being determined by the sense of the rotation. With this convention, the $\pm 1 / 4$ defects at the vertex corners are the weakest admissible defects and are referred to as trivial cases in \autoref{sec:nontrivial}.

For arbitrary constants $d _{1}, \cdots, d _{4}$, the solution of \autoref{eq:theta_EL} and \autoref{eq:theta_EL_bdry} can be written in terms of a function $f$, where $f$ is a solution of the following boundary-value problem,
\begin{align}
  \Delta f \( x, y; a, b \) & = 0, \\
  f \( x, 0; a, b \) & = 1, \\
  f \( a, y; a, b \) = f \( x, b; a, b \) = f \( 0, y; a, b \) & = 0.
\end{align}
The solution can be computed from \autoref{eq:prop3_f} in the limit $d / a \to 0$,
\begin{align}
  f \( x, y; a, b \)
  & = \sum _{k ~ \mathrm{odd}} \frac{4}{k \pi} \sin \frac{k \pi x}{a} \frac{\sinh \( k \pi \( b - y \) / a \)}{\sinh \( k \pi b / a \)},
\end{align}
with the symmetry $f \( x, y; a, b \) = f \( a - x, y; a, b \)$.
Then
\begin{align}
 \theta \( x, y \) & = d _{1} f \( x, y; a, b \)+ d _{2} f \( y, a - x; b, a \) + d _{3} f \( x, b - y; a, b \) + d _{4} f \( y, x; b, a \).
\end{align}
\begin{table}
  \centering
  \begin{tabular} {|c|c|c|c|c|c|}
    \hline
    state & shape & $d _{1}$ & $d _{2}$ & $d _{3}$ & $d _{4}$ \\
    \hline
    D1 & $\diagup$   & $0$ & $+\pi / 2$ & $0$ & $+\pi / 2$ \\
    \hline
    D2 & $\diagdown$ & $0$ & $-\pi / 2$ & $0$ & $-\pi / 2$ \\
    \hline
    R1 & $\subset$   & $0$ & $-\pi / 2$ & $-\pi$ & $-\pi / 2$ \\
    \hline
    R2 & $\supset$   & $0$ & $+\pi / 2$ & $+\pi$ & $+\pi / 2$ \\
    \hline
    R3 & $\cap$      & $0$ & $-\pi / 2$ & $0$ & $+\pi / 2$ \\
    \hline
    R4 & $\cup$      & $0$ & $+\pi / 2$ & $0$ & $-\pi / 2$ \\
    \hline
  \end{tabular}
  \caption{Boundary Condition of solutions}
  \label{tab:OF_boundarycondition}
\end{table}

The D1 state corresponds to
\begin{align}
  \theta _{\textrm{D1}} \( x, y \) = \frac{\pi}{2} \( f \( y, a - x; b, a \) + f \( y, x; b, a \) \), \quad \( x, y \) \in \( 0, a \) \times \( 0, b \), 
\end{align}
with the symmetries 
\begin{align}
  \theta _{\textrm{D1}} \( x, y \) = \theta _{\textrm{D1}} \( a - x, y \), \theta _{\textrm{D1}} \( x, y \) = \theta _{\textrm{D1}} \( x, b - y \).
\end{align}
Referring to \autoref{tab:OF_boundarycondition}, the R3 solution corresponds to
\begin{align}
  \theta _{\textrm{R3}} \( x, y \) = \frac{\pi}{2} \( - f \( y, a - x; b, a \) + f \( y, x; b, a \) \), \quad \( x, y \) \in \( 0, a \) \times \( 0, b \).
\end{align}
with the symmetries 
\begin{align}
  \theta _{\textrm{R3}} \( x, y \) = - \theta _{\textrm{R3}} \( a - x, y \), \theta _{\textrm{R3}} \( x, y \) = \theta _{\textrm{R3}} \( x, b - y \).
\end{align}
In both cases, it suffices to consider the solution of \autoref{eq:theta_EL} on the quadrant $\lsb 0, a / 2 \rsb \times \lsb 0, b / 2 \rsb$. The D1 solution is subject to the boundary conditions,
\begin{align}
  \theta _{\textrm{D1}} \( x, 0 \) & = 0, \\
  \theta _{\textrm{D1}} \( 0, y \) & = \pi / 2,
\end{align}
whereas the R3 solution is subject to,
\begin{align}
  \theta _{\textrm{R3}} \( x, 0 \) & = 0, \\
  \theta _{\textrm{R3}} \( a / 2, y \) & = 0, \\
  \theta _{\textrm{R3}} \( 0, y \) & = \pi / 2.
\end{align}
Thus, the R3 solution is a minimizer of the Dirichlet energy
\begin{align}
  \int _{\lsb 0, a / 2 \rsb \times \lsb 0, b / 2 \rsb} \lmdl \nabla \theta \( x, y \) \rmdl ^{2} \, \textrm{d} x \textrm{d} y,
\end{align}
in a smaller space than the D1 solution, which is a minimizer of the Dirichlet energy in a larger space, and hence the D1 solutions cannot have higher energies than the competing R3 solutions. This gives a simple explanation of the fact that diagonal states are observed more frequently in experiments (see \cite{lewis2014colloidal}, \cite{tsakonas2007multistable}) since they have lower energies than their rotated counterparts, without the need for any elaborate calculations.

\section{Bifurcation Diagrams as a Function of \texorpdfstring{$\epsilon$}{epsilon} and Rectangular Anisotropy}
\label{sec:bifurcation}

In \cite{robinson2017molecular}, the authors extensively discuss the reduced LdG equilibria on a square domain for low temperatures, as a function of the square length $L$.
They numerically demonstrate that the WORS (featured by a vanishing $\mathbf{Q}$-tensor along the square diagonals; see \autoref{eq:Q_s_n}) is the unique solution for $L$ small enough, its bifurcation into stable diagonal branches as $L$ increases, a further bifurcation into two unstable BD branches (defined as solutions with $Q _{12} = 0$ but without the property of $Q _{11} = Q _{12} = 0$ at the square centre) and then a tracking of how the unstable BD branches bifurcate into unstable rotated solutions and the rotated states gain stability as $L$ increases. For $L$ large enough, there are six distinct stable solution branches---the D1 and D2 solutions corresponding to the two different diagonal solutions and four energetically degenerate R1, R2, R3, R4 solutions on a square. The rotated solutions are related to each other by a $\pi / 2$ rotation on a square and are hence, energetically degenerate. As expected, we lose the degeneracy as soon as we break the symmetry of the square.

We follow the same paradigm on rectangular domains $\tilde{\Omega}$, to track solution branches as a function of the rescaled parameter $\epsilon$ (which is inversely proportional to the rectangular size) and the geometrical anisotropy, measured in terms of $a$ for fixed $b = 1$. We use arc-continuation methods \cite{kelley2018numerical} to track different solution branches and it is possible that this method does not locate all solution branches, or misses some high energy unstable branches. However, this is a standard and well-accepted method for computing bifurcation diagrams for solutions of systems of partial differential equations.

We numerically find eight different kinds of solutions on a rectangle ---the D1, D2, R1, R2, R3, R4 solutions as enumerated in the \autoref{eq:boundary_g} along with two solutions (only relevant for large $\epsilon$), labelled as BD1 and BD2. These solutions are described by their corresponding $\tilde{\mathbf{Q}}$-tensors, e.g., $\tilde{\mathbf{Q}} _{\mathrm{D1}}$, $\tilde{\mathbf{Q}} _{\mathrm{D2}}$, $\ldots$ for a given value of $\epsilon$. The BD1 and BD2 solutions are special since $\( \tilde{Q} _{\mathrm{BD1}} \) _{12} = \( \tilde{Q} _{\mathrm{BD2}} \) _{12} \equiv 0 $ which implies that the corresponding $\uvec{n}$ is either $\uvec{n} = \( 1, 0 \)$ or $\uvec{n} = \( 0, 1 \)$ (see \autoref{eq:Q_s_n}) everywhere in the rectangular interior. BD1 refer to BD solutions that have $\uvec{n} = \( 0, 1 \)$ at the centre of the rectangle with $\uvec{n} = \( 1, 0 \)$ near the horizontal edges, $y = 0$ and $y = b$. BD2 refer to BD solutions that have $\uvec{n} = \( 1, 0 \)$ at the centre of the domain with $\uvec{n} = \( 0, 1 \)$ near the vertical edges, $x = 0$ and $x = a$. For $a = b$, the solutions BD1 and BD2 are energetically degenerate whereas for $a > b$, one can heuristically see that BD1 is unfavourable compared to BD2 since the transition layers for BD1 are located along the longer edges, $y = 0$ and $y = b$. We point out that BD2 is the limiting profile, $\tilde{\mathbf{Q}} _{\infty}$ on a rectangle $\tilde{\Omega}$ with $a > b$ discussed in \autoref{sec:strong}. On similar grounds, we lose the degeneracy between the rotated states and the R1 and R2 states have higher energies than the R3, R4 states for $a > b$. We also study the stability of the solutions, and use prefixes `s' and `u' for stable and unstable solutions respectively. We test the stability ($L ^{2}$-stability) of a solution by computing the smallest eigenvalue of the second variation of the discrete LdG free energy functional \cite{wiki:Hessian}, and if the smallest eigenvalue is positive, the solution is locally stable.

We use the notation $E _{\textrm{c}} \( \epsilon, a \)$ to denote the Landau-de Gennes energy of the solution in class `c' at given $\epsilon$ and geometrical aspect ratio $\delta = \frac{a}{b} = a$ (since $b = 1$ in our simulations).
We plot the energies of the solutions of the Euler-Lagrange equation with respect to the parameter $\epsilon$, see \autoref{fig:Bifurcation__LC_bifurcation_N4_h64_64_120945_E}, \autoref{fig:Bifurcation__LC_bifurcation_N4_h80_64_120945_E} and \autoref{fig:Bifurcation__LC_bifurcation_N4_h96_64_120945_E}. For small values of $\epsilon$ and $a > b$, there are three competing stable solutions, sD1, sR2 and sR3, and the energies are ordered as $E _{\textrm{sD1}} \( \epsilon, a \) < E _{\textrm{sR3}} \( \epsilon, a \) \leq E _{\textrm{sR2}} \( \epsilon, a \)$, which is consistent with the results in \cite{lewis2014colloidal}.

We track three different solution pathways, with three distinct initial conditions sD1, sR3 and sR2 respectively, for small values of $\epsilon$. We discuss each pathway separately.
As $\epsilon$ increases, the solution sD1 transitions to sBD2, which is the unique limiting solution described in \autoref{prop:1}; we do not observe any unstable solutions in this pathway and the qualitative features seem to be independent of the geometrical anisotropy $a$. We define the parameter, $\epsilon _{\textrm{sD1} \to \textrm{sBD2}} \( a \)$, be the value of $\epsilon$ where this transition occurs, defined by $\( \tilde{Q} _{\textrm{sBD2}} \) _{12} \equiv 0 $. As $a$ increases from $1$, $\epsilon _{\textrm{sD1} \to \textrm{sBD2}} \( a \)$ decreases. We illustrate this transition in \autoref{fig:sD1_sBD2}.

\begin{figure}[ht]
  \centering
  \includegraphics[width = 0.5\textwidth] {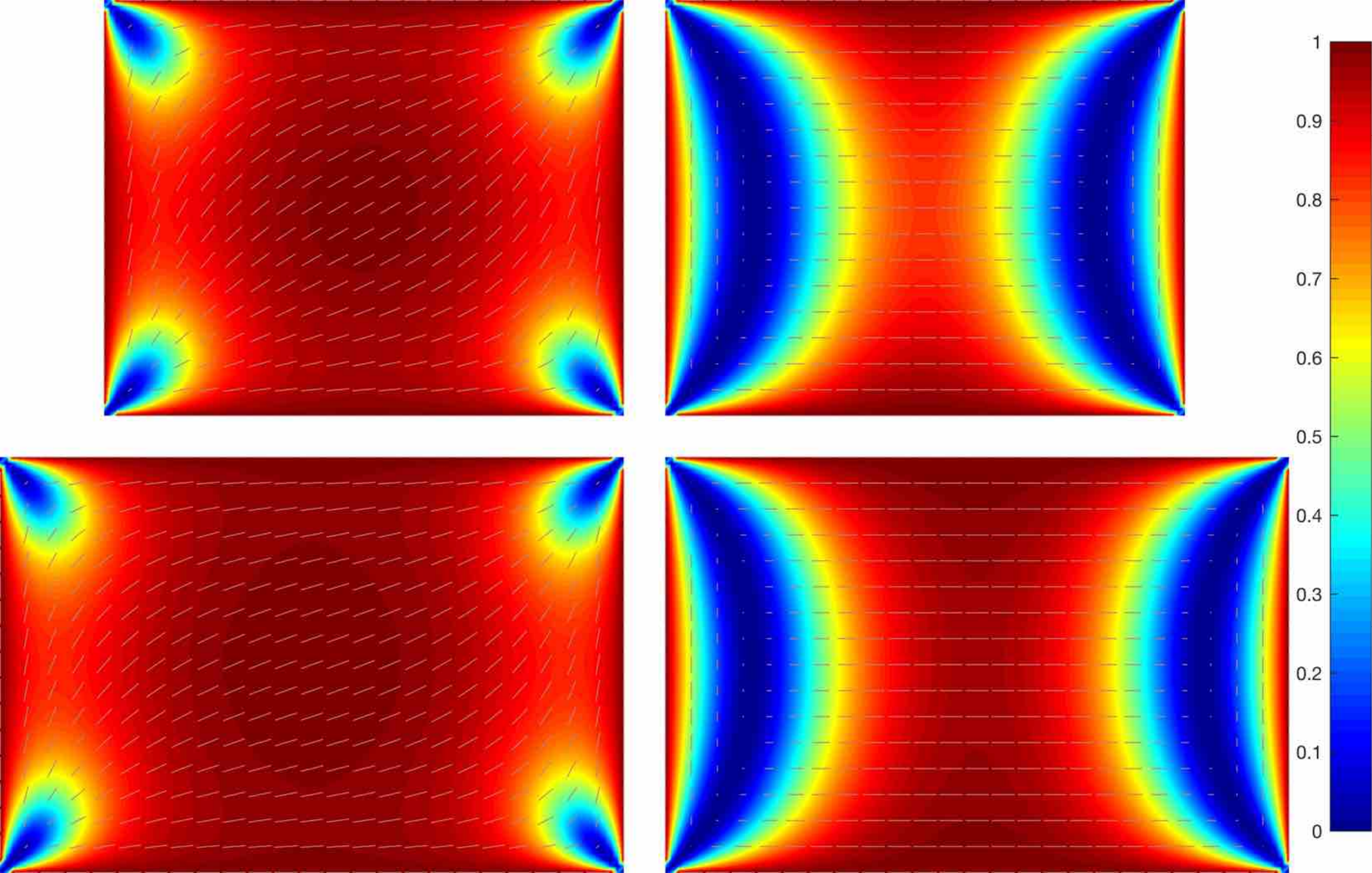} \\
  \caption{Top, from left to right: sD1 and sBD2 at domain size $1.25 \times 1$, which are marked as `A' and `B' in \autoref{fig:Bifurcation__LC_bifurcation_N4_h80_64_120945_E} and \autoref{fig:Bifurcation__LC_bifurcation_N4_h80_64_120945_Q}. Bottom, from left to right: sD1 and sBD2 at domain size $1.5 \times 1$, which are marked as `A' and `B' in \autoref{fig:Bifurcation__LC_bifurcation_N4_h96_64_120945_E} and \autoref{fig:Bifurcation__LC_bifurcation_N4_h96_64_120945_Q}. Mesh spacing is $h = 1 / 64$. We fix $\( \tilde{Q} _{11}, \tilde{Q} _{12} \) = \tilde{\mathbf{g}} _{d}$ on the boundary with $d = 0.03$. The colour bar represents the value of $\tilde{s} _{\epsilon} ^{2} = \tr{\tilde{\mathbf{Q}} _{\epsilon} ^{2}} / 2$.}
  \label{fig:sD1_sBD2}
\end{figure}

The second pathway has the competing stable rotated solution, sR3, as initial condition, and we gradually increase $\epsilon$ and track the corresponding changes. The second pathway is different from the first pathway described above. We observe the following sequence of transitions from sR3 to uR3, to uBD2, and finally to sBD2 and the qualitative features are independent of $a$. We can similarly define the parameters $\epsilon _{\textrm{sR3} \to \textrm{uR3}} \( a \) < \epsilon _{\textrm{uR3} \to \textrm{uBD2}} \( a \) < \epsilon _{\textrm{uBD2} \to \textrm{sBD2}} \( a \)$. As $a$ increases from $1$, $\epsilon _{\textrm{sR3} \to \textrm{uR3}} \( a \)$ and $\epsilon _{\textrm{uR3} \to \textrm{uBD2}} \( a \)$ increase, and $\epsilon _{\textrm{uBD2} \to \textrm{sBD2}} \( a \) \equiv \epsilon _{\textrm{sD1} \to \textrm{sBD2}} \( a \)$ decreases. As $a \to \infty$, the differences among these three critical parameters tend to zero. We illustrate this sequence of structural transitions in \autoref{fig:sR3_sBD2}.

\begin{figure}[ht]
  \centering
  \includegraphics[width = \textwidth] {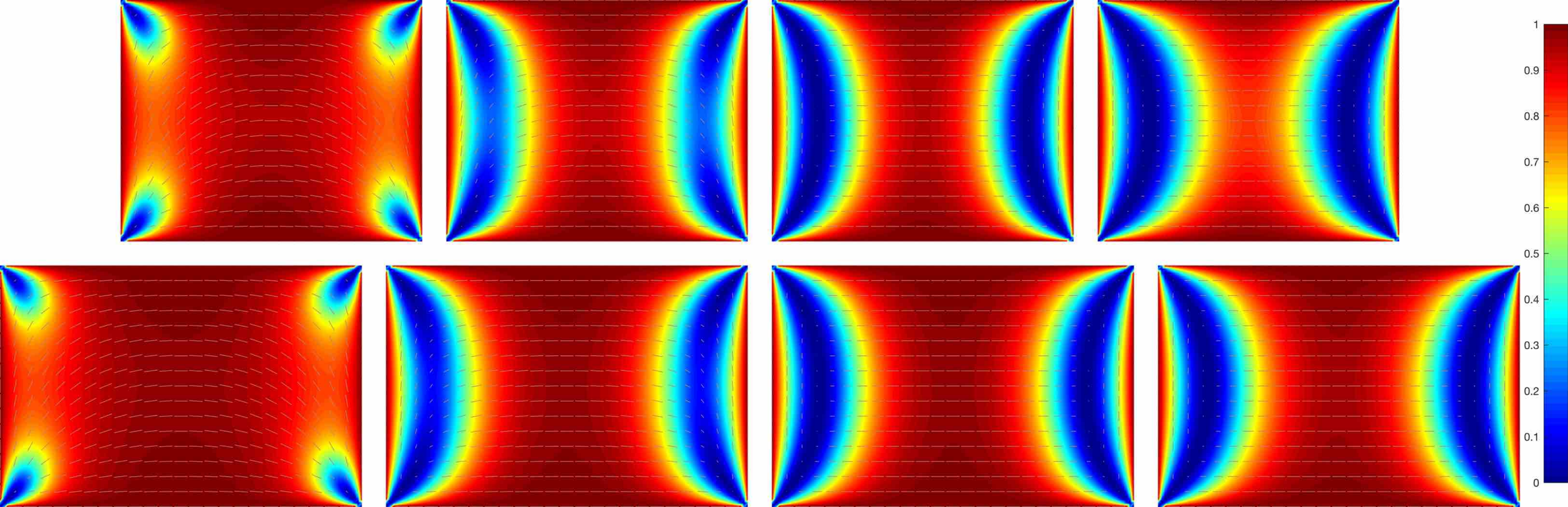} \\
  \caption{Top, from left to right: sR3, uR3, uBD2 and sBD2 at domain size $1.25 \times 1$, which are marked as `G', `H', `I' and `B' in \autoref{fig:Bifurcation__LC_bifurcation_N4_h80_64_120945_E} and \autoref{fig:Bifurcation__LC_bifurcation_N4_h80_64_120945_Q}. Bottom, from left to right: sR3, uR3, uBD2 and sBD2 at domain size $1.5 \times 1$, which are marked as `G', `H', `I' and `B' in \autoref{fig:Bifurcation__LC_bifurcation_N4_h96_64_120945_E} and \autoref{fig:Bifurcation__LC_bifurcation_N4_h96_64_120945_Q}. Mesh spacing is $h = 1 / 64$. We fix $\( \tilde{Q} _{11}, \tilde{Q} _{12} \) = \tilde{\mathbf{g}} _{d}$ on the boundary with $d = 0.03$. The colour bar represents the value of $\tilde{s} _{\epsilon} ^{2} = \tr{\tilde{\mathbf{Q}} _{\epsilon} ^{2}} / 2$.}
  \label{fig:sR3_sBD2}
\end{figure}

The third pathway is quite different, for which we use sR2 as initial condition for $\epsilon$ small enough and then gradually increase the parameter $\epsilon$. In fact, the qualitative features of the pathway also depend on the value of $a > 1$.
At small $a$, e.g. $a = 1.25$, the solution sR2 changes to uR2, then to uBD1, and finally the solution branch disappears at an end point. We can similarly define the parameters $\epsilon _{\textrm{sR2} \to \textrm{uR2}} \( a \) < \epsilon _{\textrm{uR2} \to \textrm{uBD1}} \( a \) < \epsilon _{\textrm{end}} \( a \)$, such that we cannot find further solutions on this branch by arc-continuation methods for $\epsilon > \epsilon _{\textrm{end}} \( a \)$. At intermediate $a$, e.g., $a = 1.5$, the solution sR2 transitions to uR2 and the solution branch disappears on increasing $\epsilon$, and we have $\epsilon _{\textrm{sR2} \to \textrm{uR2}} \( a \) < \epsilon _{\textrm{end}} \( a \)$. There is no intermediate transition to uBD1 for intermediate values of $a > 1$. At larger $a$, we observe the termination of the branch at sR2. As $a$ increases from $1$, $\epsilon _{\textrm{end}} \( a \) < \epsilon _{\textrm{sD1} \to \textrm{sBD2}} \( a \)$ decreases. We illustrate these structural transition pathways in \autoref{fig:sR2_uBD1}, for different values of $a$.

We plot bifurcation diagrams for the reduced LdG equilibria for three values of $a = 1, 1.25, 1.5$ respectively, following the same measures as in \cite{robinson2017molecular}. The figures are labelled as  \autoref{fig:Bifurcation__LC_bifurcation_N4_h64_64_120945_Q}, \autoref{fig:Bifurcation__LC_bifurcation_N4_h80_64_120945_Q} and \autoref{fig:Bifurcation__LC_bifurcation_N4_h96_64_120945_Q}.
These plots illustrate the distinct solution pathways and changes in the number of solutions as $\epsilon$  decreases, and the sensitivity of the solution landscape to the aspect ratio $a$, although more exhaustive studies are needed to this effect. We use the quantities $\tilde{\Omega} ^{-1} \int _{\tilde{\Omega}} \tilde{Q} _{11} ^{2}$, $\tilde{\Omega} ^{-1} \int _{\tilde{\Omega}} \tilde{Q} _{12} ^{2}$ and $\tilde{\Omega} ^{-1} \int _{\tilde{\Omega}} s ^{2}$ in \autoref{fig:Bifurcation__LC_bifurcation_N4_h72_64_300121_largeepsilon_Q12}, from which we deduce that the structural transition from sD1 to sBD2 is a first order structural transition i.e. discontinuous as a function of $\epsilon$.

The pertinent question is---how are the solution landscapes on a rectangle different from the solution landscapes on a square? There are no striking differences except perhaps that the sR2 solution pathway is not connected to the sD1 and sR3 pathways directly for $a > 1$. Similarly, we do not find a sBD1 state for a rectangle even for smaller values of $a$, for which the energetic penalty of the transition layers along the longer horizontal edges is lower. It is interesting that sD1 and uBD2 transition to sBD2 at the same value of $\epsilon$ on two distinct solution pathways. It is worth noting that BD states are always unstable on a square and this is an illuminating example of how geometrical anisotropy can stabilise BD states on a rectangle. Further, the solution pathways do not appear to be highly sensitive to $a$, although this could be a limitation of the numerical methods. Further, as $a$ increases from $1$, and for small $\epsilon$, $E _{\textrm{sD1}} \( \epsilon, a \)$ and $E _{\textrm{sR2}} \( \epsilon, a \)$ increase, while $E _{\textrm{sR3}} \( \epsilon, a \)$ decreases. When $a = 1$, $E _{\textrm{sR3}} \( \epsilon, a \) = E _{\textrm{sR2}} \( \epsilon, a \)$. As $a \to \infty$, $E _{\textrm{sR3}} \( \epsilon, a \) \to E _{\textrm{sD1}} \( \epsilon, a \)$, for any fixed and small $\epsilon$, which agrees with the results in \cite{lewis2014colloidal}. An open question pertains to how a system would transition from the sR2 solution  to the sBD2 solution as $\epsilon$ increases, since sBD2 is the unique solution for $\epsilon$ large enough? 
\begin{figure}[ht]
  \centering
  \includegraphics[width = \textwidth] {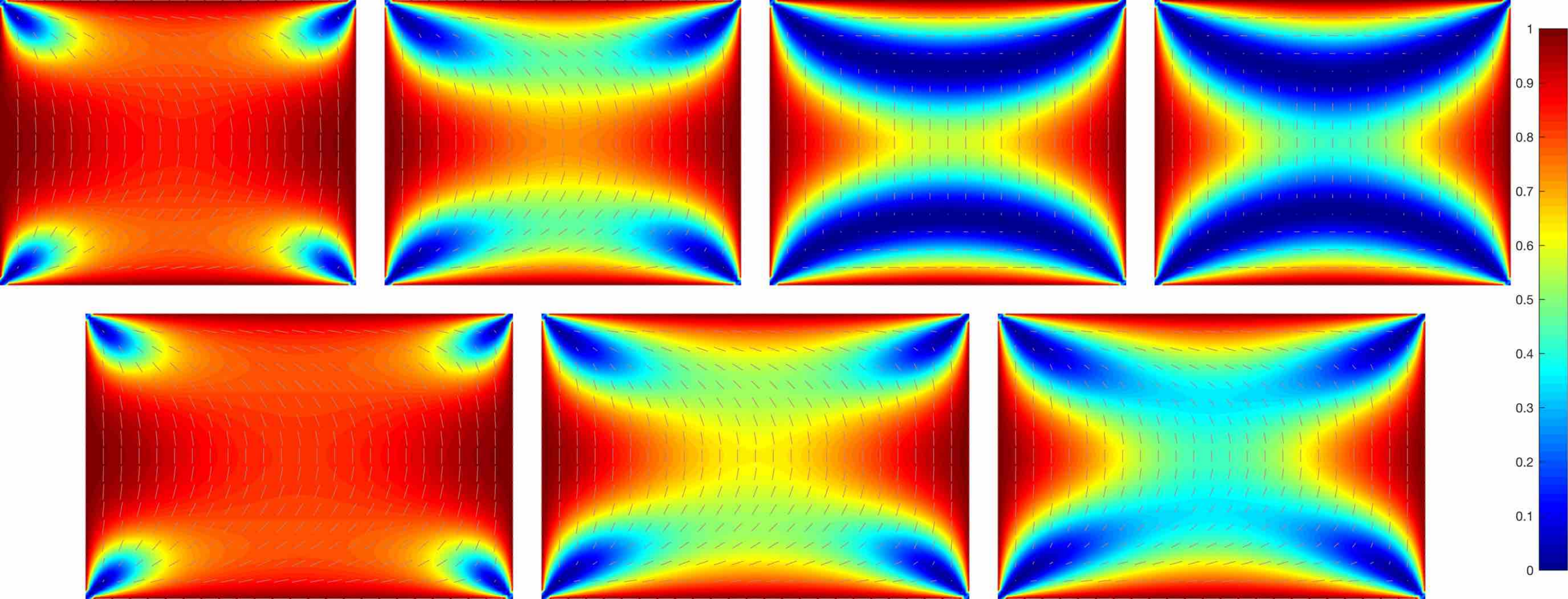} \\
  \caption{Top, from left to right: sR2, uR2, uBD1 and end (uBD1) at domain size $1.25 \times 1$, which are marked as `C', `D', `E' and `F' in \autoref{fig:Bifurcation__LC_bifurcation_N4_h80_64_120945_E} and \autoref{fig:Bifurcation__LC_bifurcation_N4_h80_64_120945_Q}. Bottom, from left to right: sR2, uR2 and end (uR2) at domain size $1.5 \times 1$, which are marked as `C', `D' and `F' in \autoref{fig:Bifurcation__LC_bifurcation_N4_h96_64_120945_E} and \autoref{fig:Bifurcation__LC_bifurcation_N4_h96_64_120945_Q}. Mesh spacing is $h = 1 / 64$. We fix $\( \tilde{Q} _{11}, \tilde{Q} _{12} \) = \tilde{\mathbf{g}} _{d}$ on the boundary with $d = 0.03$. The colour bar represents the value of $\tilde{s} _{\epsilon} ^{2} = \tr{\tilde{\mathbf{Q}} _{\epsilon} ^{2}} / 2$.}
  \label{fig:sR2_uBD1}
\end{figure}

\begin{figure}[ht]
  \centering
  \begin{subfigure}[t]{0.49\textwidth}
    \includegraphics[width = \textwidth] {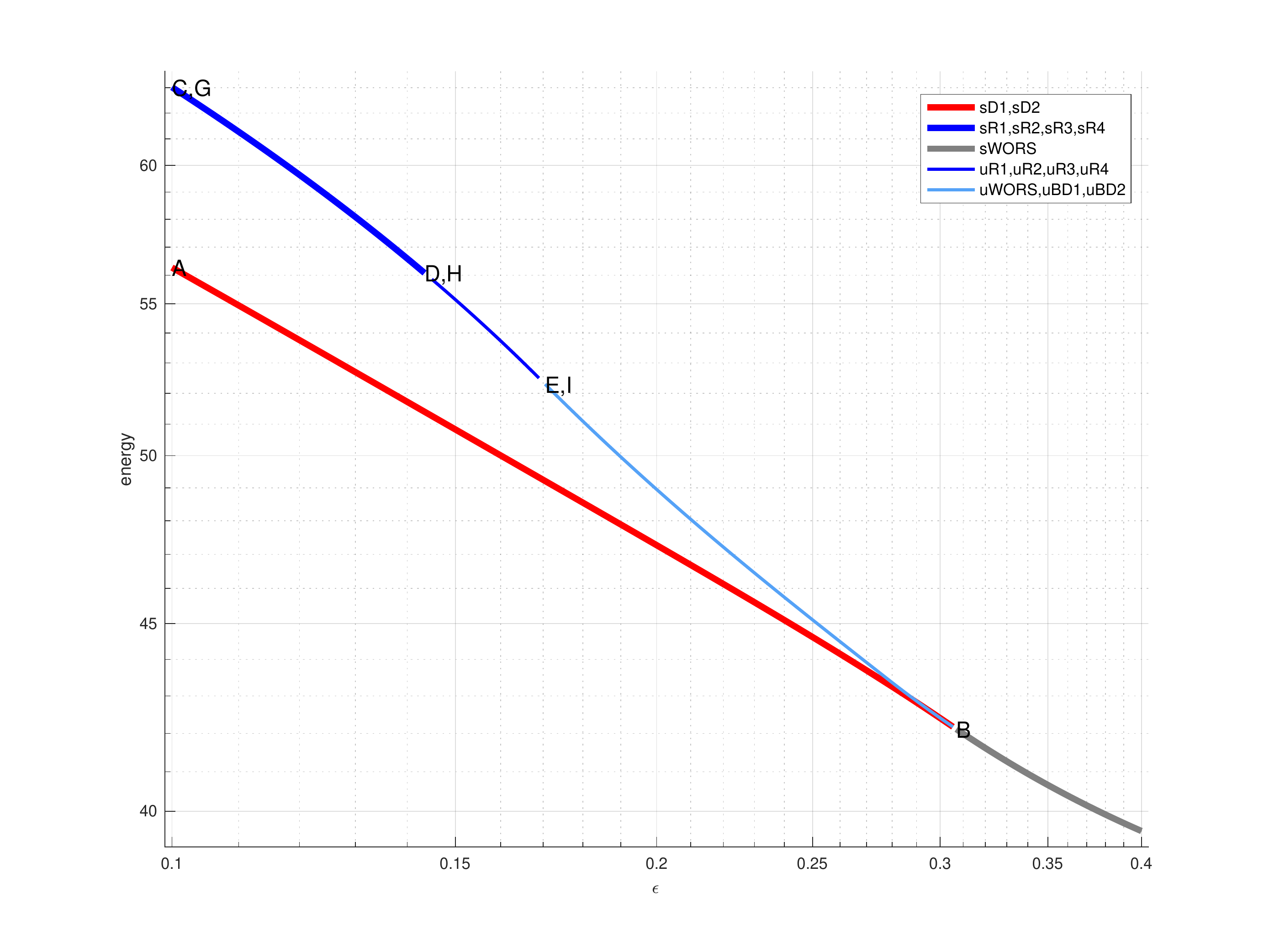} \\
    \caption{}
    \label{fig:Bifurcation__LC_bifurcation_N4_h64_64_120945_E}
  \end{subfigure}
  \begin{subfigure}[t]{0.49\textwidth}
    \includegraphics[width = \textwidth] {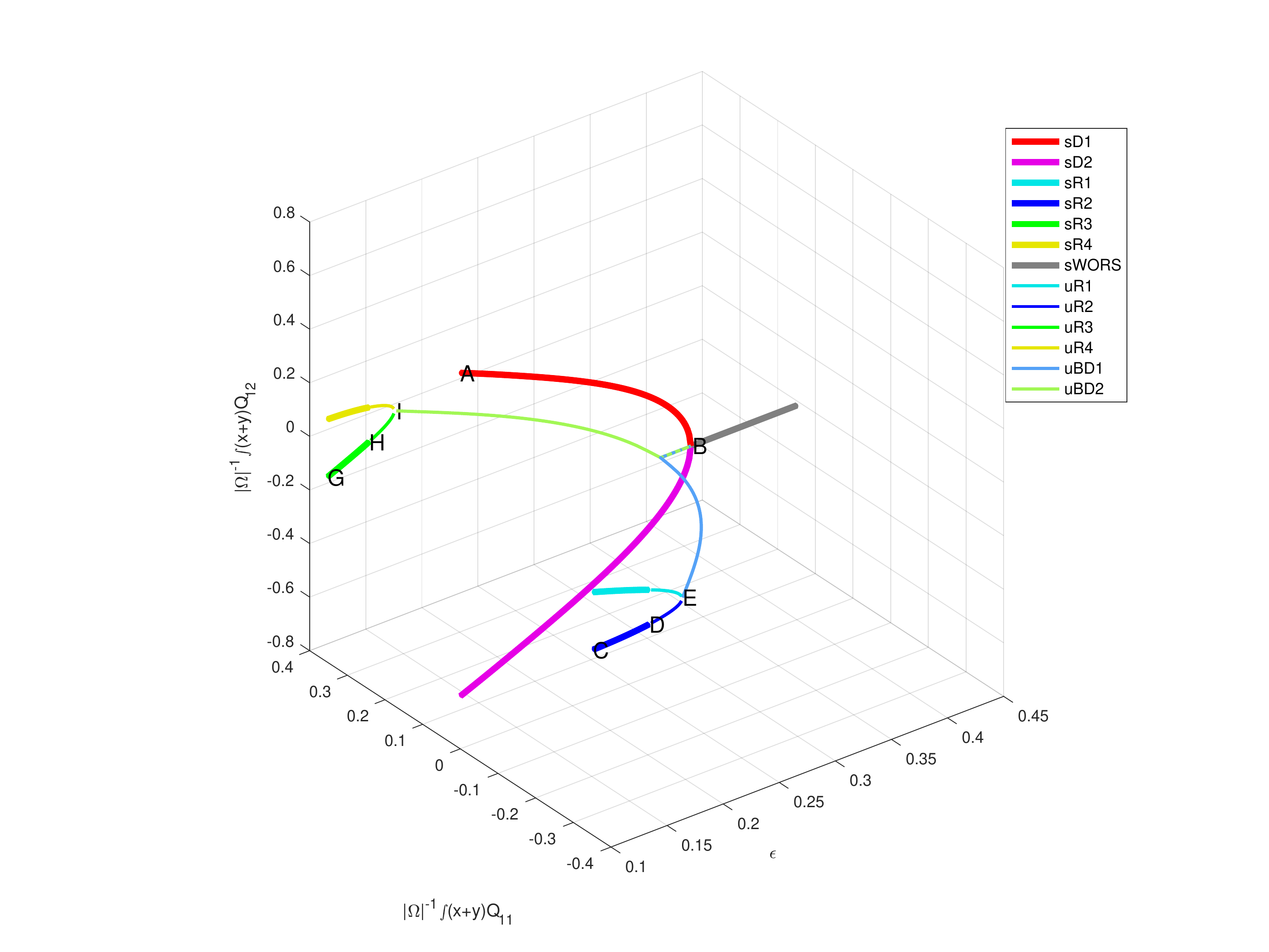} \\
    \caption{}
    \label{fig:Bifurcation__LC_bifurcation_N4_h64_64_120945_Q}
  \end{subfigure}
  \\
  \begin{subfigure}[t]{0.49\textwidth}
    \includegraphics[width = \textwidth] {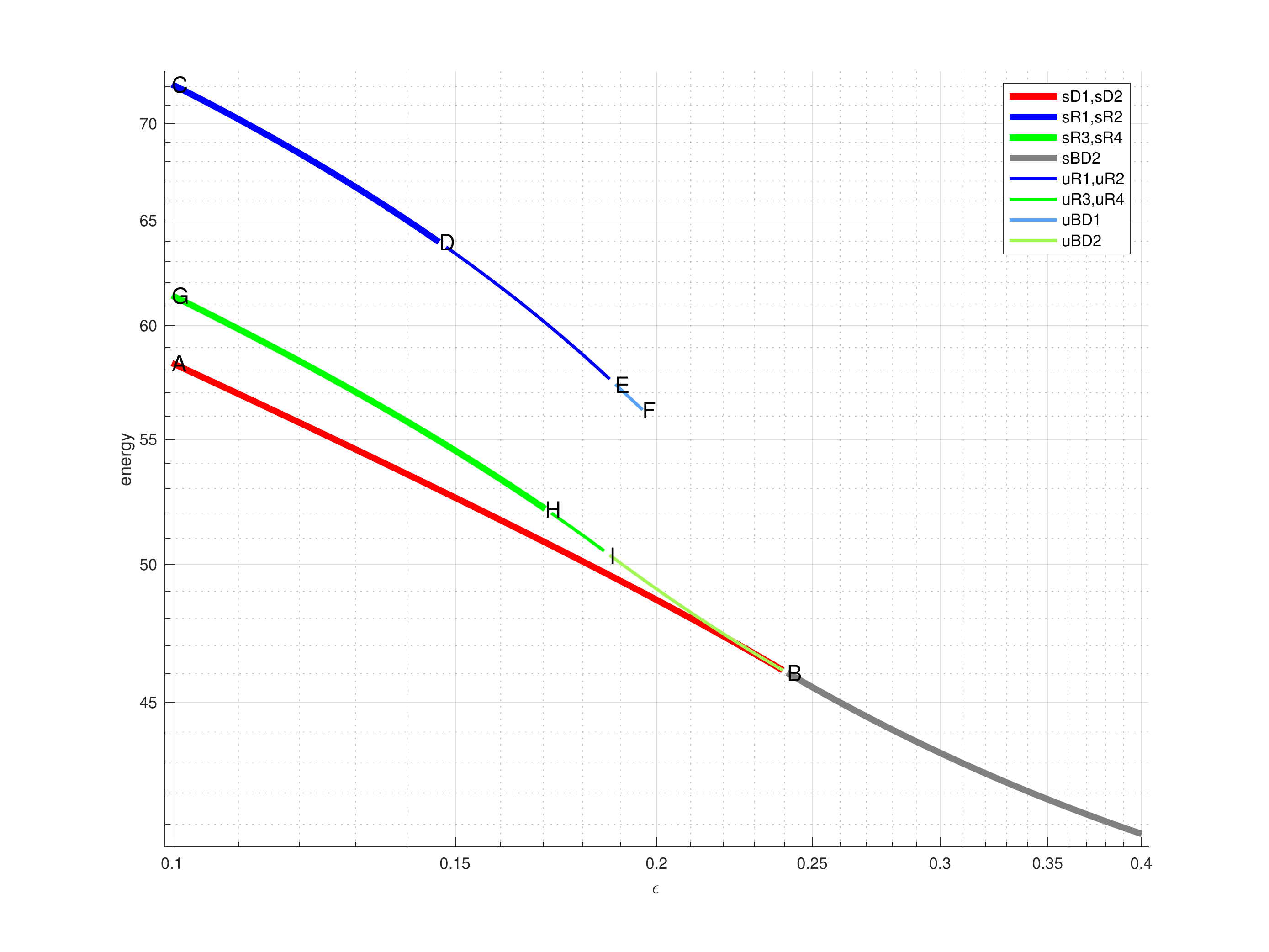} \\
    \caption{}
    \label{fig:Bifurcation__LC_bifurcation_N4_h80_64_120945_E}
  \end{subfigure}
  \begin{subfigure}[t]{0.49\textwidth}
    \includegraphics[width = \textwidth] {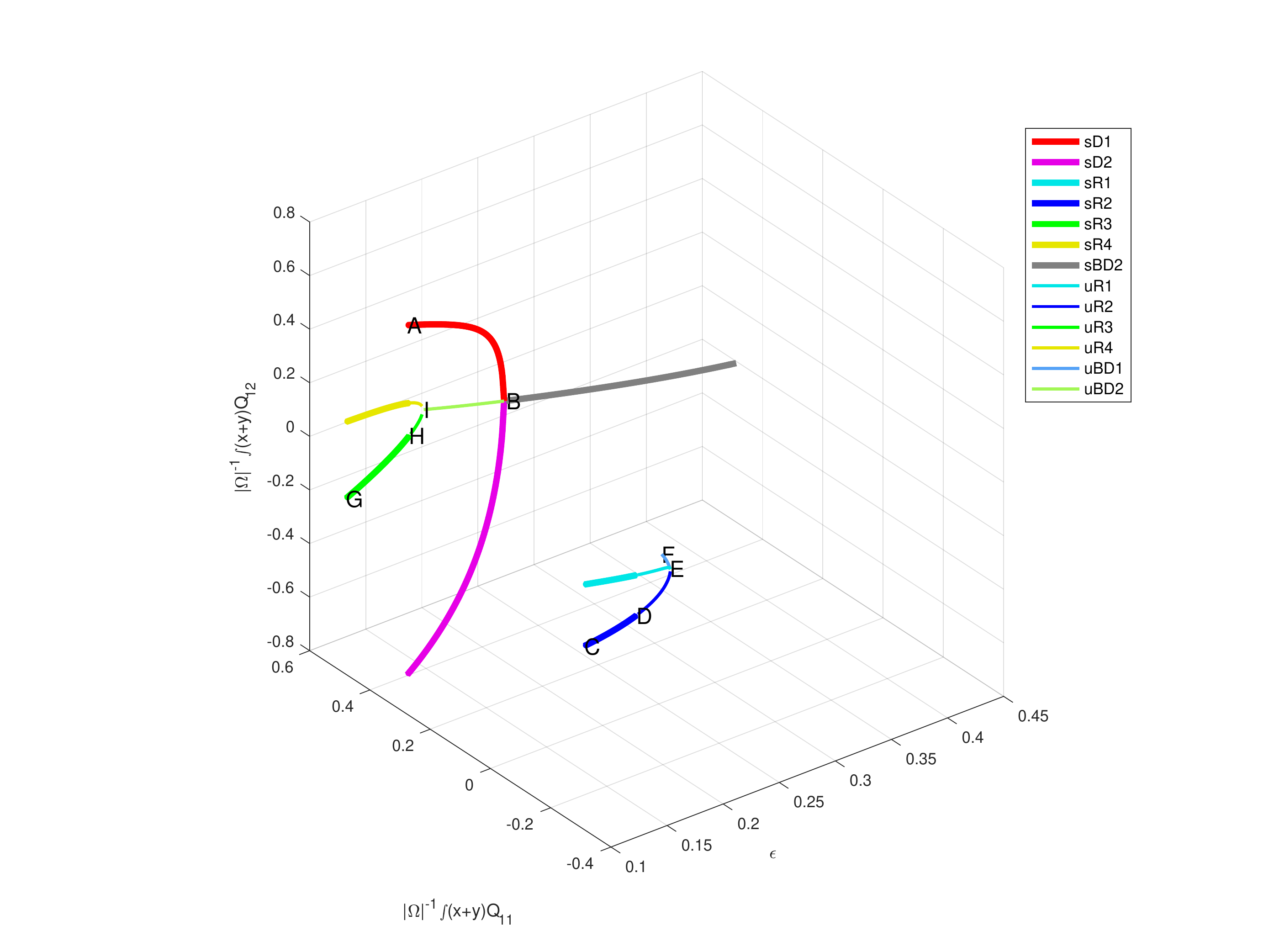} \\
    \caption{}
    \label{fig:Bifurcation__LC_bifurcation_N4_h80_64_120945_Q}
  \end{subfigure}
  \\
  \begin{subfigure}[t]{0.49\textwidth}
    \includegraphics[width = \textwidth] {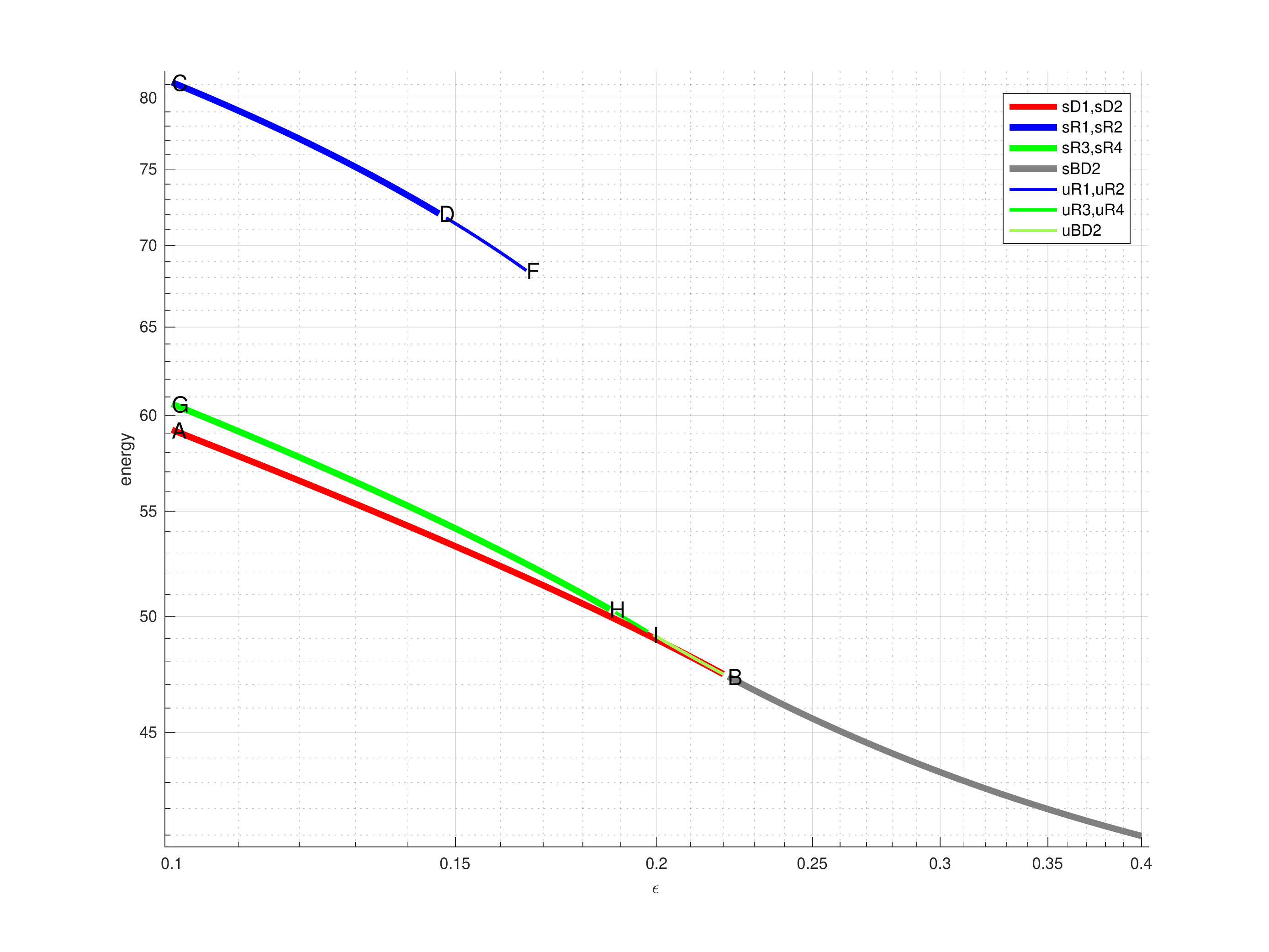} \\
    \caption{}
    \label{fig:Bifurcation__LC_bifurcation_N4_h96_64_120945_E}
  \end{subfigure}
  \begin{subfigure}[t]{0.49\textwidth}
    \includegraphics[width = \textwidth] {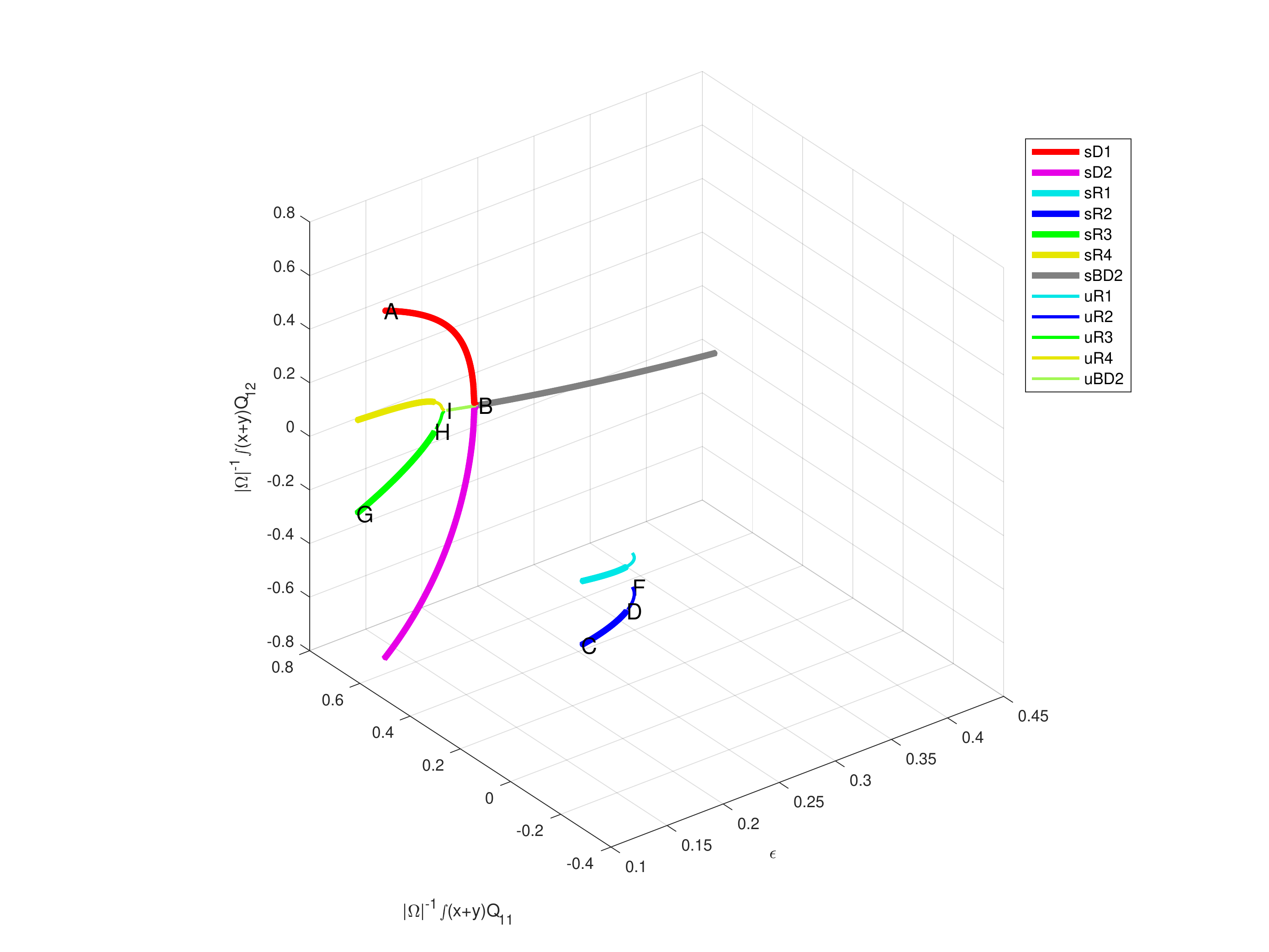} \\
    \caption{}
    \label{fig:Bifurcation__LC_bifurcation_N4_h96_64_120945_Q}
  \end{subfigure}
  \caption{Bifurcations in the Landau-de Gennes (LdG) model. Mesh spacing is $h = 1 / 64$. Left column: the LdG free energy versus $\epsilon$. Right column: relation among $\epsilon$, $\fint _{\tilde{\Omega}} \( \tilde{x} + \tilde{y} \) \tilde{Q} _{11} \( \tilde{x}, \tilde{y} \)$ and $\fint _{\tilde{\Omega}} \( \tilde{x} + \tilde{y} \) \tilde{Q} _{12} \( \tilde{x}, \tilde{y} \)$. Top row: domain size is $1 \times 1$. Middle row: domain size is $1.25 \times 1$. Bottom row: domain size is $1.5 \times 1$. See snapshots marked by capital letters in \autoref{fig:sD1_sBD2}, \autoref{fig:sR3_sBD2} and \autoref{fig:sR2_uBD1}.}
  \label{fig:Bifurcation}
\end{figure}

\begin{figure}[ht]
  \centering
  \includegraphics[width = 0.5\textwidth] {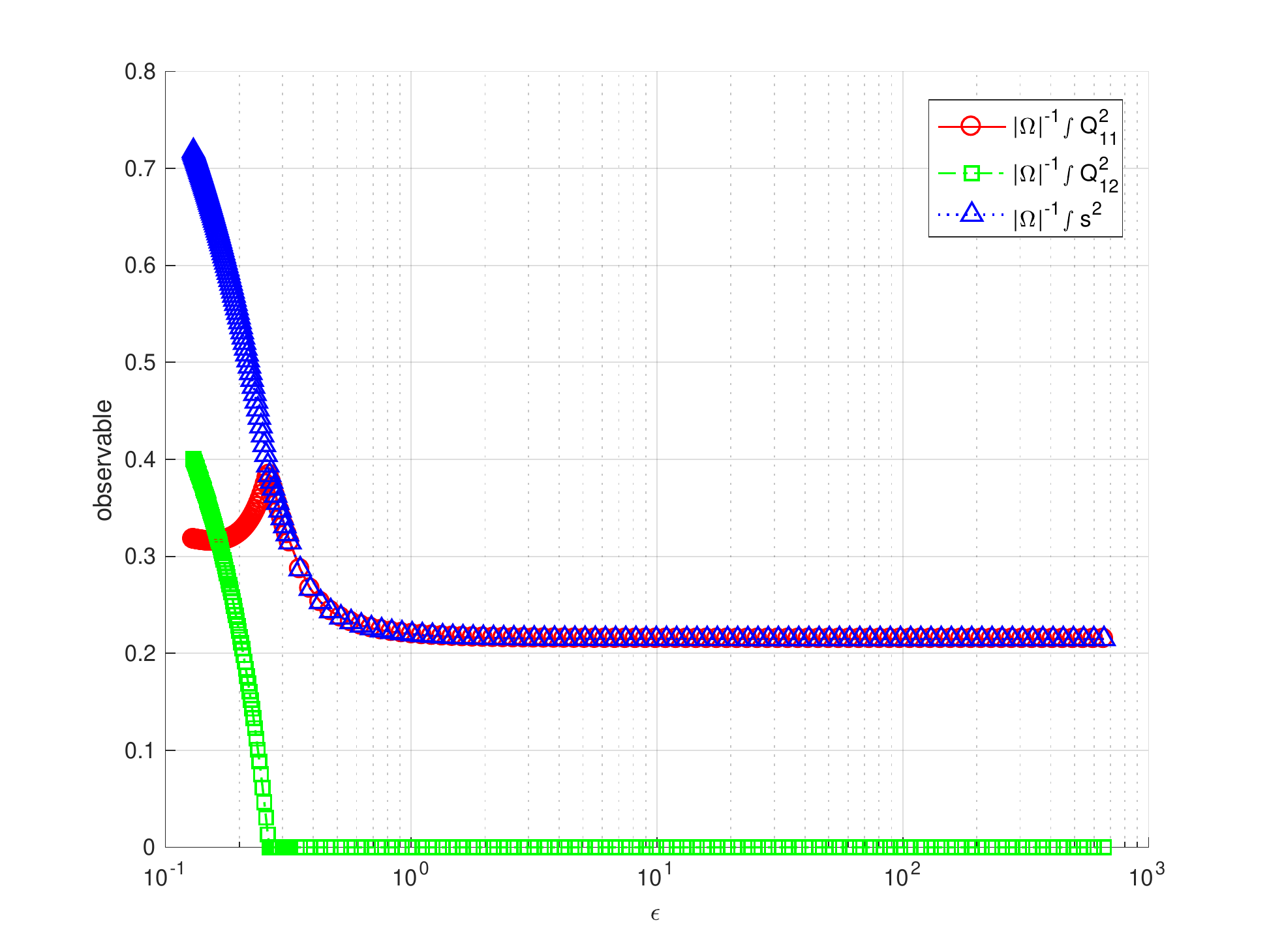} \\
  \caption{Characteristics of sD1 and sBD2 solutions as a function of $\epsilon$. Domain size $1.125 \times 1$, mesh spacing $h = 1 / 64$.}
  \label{fig:Bifurcation__LC_bifurcation_N4_h72_64_300121_largeepsilon_Q12}
\end{figure}

\section{Relaxation Mechanisms for Non-trivial Topologies}
\label{sec:nontrivial}

We can classify nematic equilibria on a rectangular domain with tangent boundary conditions in terms of their structural details near the rectangular vertices. More precisely, a nematic state, in our framework, is described by a 2D LdG $\mathbf{Q}$-tensor. In terms of the corresponding director, $\bf n = \( \cos \theta, \sin \theta \)$, we necessarily need that $\theta$ is some multiple of $\pi$ on the horizontal edges and that $\theta$ is an odd multiple of $\pi / 2$ on the vertical edges. In the simplest case, the director rotates by $\pm \pi / 2$ radians between a pair of adjacent edges, as is the case with the diagonal and rotated solutions described above. We refer to these solutions as having `trivial' topologies. 
However, it is admissible for the director to rotate by $n \pi / 2$ radians between a pair of adjacent edges, for an odd integer $\lmdl n \rmdl \geq 3$, and such equilibria are deemed to have `non-trivial' topologies. The degree of a vertex is defined to be $\omega = \frac{1}{2\pi}\frac{n \pi}{2} = \frac{n}{4}$ for some odd integer $n$, which is positive if the rotation is anticlockwise along an anticlockwise-oriented arc connecting the two intersecting edges at the vertex in question, and negative otherwise.

It is natural to ask if non-trivial topologies can be realised in practice? One can conceive a three-dimensional set-up of a shallow well $\mathcal{B}$ with a rectangular cross-section $\Omega$, with Dirichlet uniaxial tangent conditions on the lateral surfaces and suitable surface anchoring energies on the top and bottom surfaces. As argued in previous sections, under physically reasonable assumptions, we can reduce to a rescaled two-dimensional problem for which we study planar profiles on a rectangle $\tilde{\Omega} = \lsb 0, a \rsb \times \lsb 0, b \rsb$ with tangent conditions for the directors on the edges. The liquid crystal molecules could be locally rotated near chosen vertices, inducing non-trivial behaviour or non-trivial topologies. We expect these non-trivial topologies to relax to trivial topologies once the local rotation is removed and the relaxation pathways will be sensitive to $\epsilon$ and $\delta = \frac{a}{b}$, with interesting consequences for the corresponding optical properties and the relaxation pathways could even offer new switching mechanisms between diagonal and rotated states. Whilst this is largely a proof of concept at this stage, we explore some relaxation mechanisms from non-trivial to trivial topologies in this section as described below.

The first step is to construct an initial condition with a non-trivial topology for our numerical solvers. As a concrete example, we numerically solve
\begin{align}
  \Delta \theta \( x, y \) = 0, \quad \( x, y \) \in \tilde{\Omega}, \\
  \theta \( x, 0 \) = d _{1} = 0, \quad \theta \( a, y \) = d _{2} = \frac{\pi}{2}, \quad \theta \( x, b \) = d _{3} = 2 \pi, \quad \theta \( 0, y \) = d _{4} = \frac{5\pi}{2}.
\end{align} 
Let $\theta _{0}$ denote the solution of the boundary-value problem above; then $\theta _{0}$ has a $5 / 4$-degree defect at $\( 0, 0 \)$, a $-3 / 4$-degree defect at $\( a, b \)$, and two $-1 / 4$-degree defects at $\( 0, b \)$ and $\( a, 0 \)$.
We refer to $\( 0, 0 \)$ and $\( a, b \)$ as non-trivial vertices and to $\( 0, b \)$ and $\( a, 0 \)$ as being trivial vertices. 
We use $\theta _{0}$ to define a $\tilde{\mathbf{Q}} _{0}$-tensor field on the rescaled rectangle $\tilde{\Omega}$, by setting $\( \tilde{Q} _{0} \) _{11} = \cos \( 2 \theta _{0} \)$, $\( \tilde{Q} _{0} \) _{12} =  \sin \( 2 \theta _{0} \)$. The boundary condition $\( \( \tilde{Q} _{0} \) _{11}, \( \tilde{Q} _{0} \) _{12} \) = \tilde{\mathbf{g}} _{d}$ on $\partial \tilde{\Omega}$ for a given $d > 0$ (sufficiently small). Then we can compute the 2D LdG free energy minimizer with $\tilde{\mathbf{Q}} _{0}$ as initial condition, using the steepest descent method \cite{cauchy1847methode}. It is clear that the numerical procedure will converge to either sBD2 (when $a > b$) for large $\epsilon$ or one of the diagonal or rotated states for small $\epsilon$ and we record the intermediate states during the energy minimization procedure to study the relaxation mechanisms as a function of $\epsilon$ and the aspect ratio $\delta$.

In the numerical implementation, we fix $d = 0.03$ for the boundary condition, and $b = 1$ . 
For $a = 1$, i.e. a square domain, we find three different relaxation modes for different ranges of $\epsilon$ in the 2D LdG model (see \autoref{eq:Qtilde} to \autoref{eq:boundary_T}). For small $\epsilon$ (for example, $\epsilon = 0.05$), the non-trivial defects with $\lmdl \omega \rmdl \geq 3 / 4$, split into multiple defects and pairs of these newly created defects merge near vertices to give us the diagonal state that connects the two non-trivial vertices (D1). More precisely, we see two $1 / 2$-defects emanate from the $5 / 4$-vertex and two $-1 / 2$-defects emanate from the $-3 / 4$-vertex, and these two pairs attract each other and annihilate near the two trivial vertices. These leaves us with two $1 / 4$-degree vertices at $\( 0, 0 \)$ and $\( a, b \)$ and the D1 state connecting them. For intermediate values of $\epsilon$ (for example, $\epsilon = 0.09$), the $-3 / 4$-defect expels a $-1 / 2$-defect into the interior and the $5 / 4$-defect expels three $+1 / 2$-defects into the interior. A $\pm 1 / 2$-pair annihilates in the bulk and the remaining two $+1 / 2$-defects move along the shorter edges to the trivial vertices. In other words, we observe defect annihilation in the bulk and get the diagonal solution with the opposite diagonal orientation (D2); for large $\epsilon$ (for example, $\epsilon = 0.4$), the initial state relaxes to the sWORS state as expected, through almost isotropic states in the bulk. See \autoref{fig:nontrivial_1} for more details.
For $a = 5$, i.e., a rectangular domain, we find three analogous relaxation modes for different values of $\epsilon$. For small $\epsilon$ (for example, $\epsilon = 0.05$), the non-trivial defects (with degree $\lmdl \omega \rmdl \geq 3 / 4$) split into multiple defects and some of these defects move towards each other in the bulk and others migrate to the nearest vertex to yield the rotated state (R3); for intermediate $\epsilon$  (for example, $\epsilon = 0.09$), we observe the defect splitting near the non-trivial vertices and the expelled defects move along the long edges to either annihilate each other or change the topology of a vertex, yielding the diagonal solution (D2); for large values of $\epsilon$ (for example, $\epsilon = 0.4$), the initial state relaxes to the sBD2 state through almost isotropic states (with $\tilde{\mathbf{Q}} = 0$) in the interior. See \autoref{fig:nontrivial_5}.

These examples are certainly not exhaustive. However, they do highlight certain generic concepts---the non-trivial vertices will relax into trivial vertices by expelling the excess charge, $\eta = \( |n| - 1 \) / 2$, as  a combination of $\pm 1 / 2$ defects in the director field. This combination is not unique, and we cannot predict the combination as a function of $\epsilon$ and $\delta$ at present. For small values of $\epsilon$ and for $\delta = 1, 5$, the transient states are ordered (with relatively high scalar order parameter) and the different relaxation mechanisms will offer different optical properties. Intermediate values of $\epsilon$ exhibit transient states with almost line defects i.e. lines of low order along the rectangular edges or in the interior, surrounded by regions of high order. As $\epsilon$ increases further, the transient states become increasingly disordered and again, these could be interesting examples of `wetted' nematic states. Further, it seems that non-trivial topologies can relax to either the diagonal or rotated states and there appears to be no clear selection mechanism. Therefore, these numerical experiments offer interesting new routes for exploiting non-trivial topologies--- to create new defects, control defect dynamics, create new wetted states and even offer new pathways between diagonal and rotated states for switching dynamics. The role of $\epsilon$ is relatively clear, both for the size of the disordered regions in the transient states (the size increases linearly with $\epsilon$ as expected and hence, defects can move more easily for smaller values of $\epsilon$) and for the choices of the final relaxed state but the role of the geometrical anisotropy is less clear from these simulations.

\begin{figure}
  \centering
  \includegraphics[width = 0.3\textwidth] {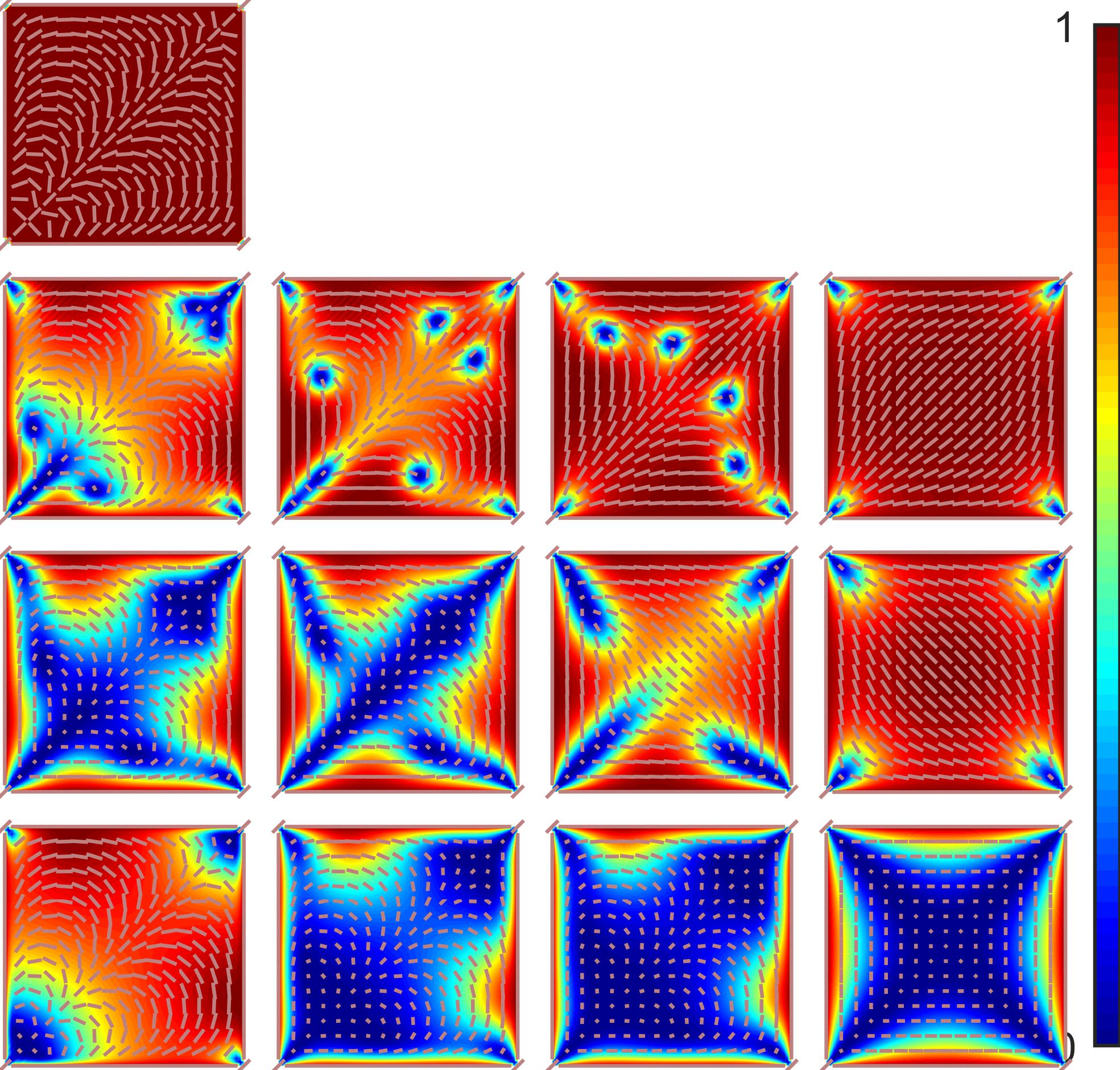} \\
  \caption{Relaxation of a non-trivial topology to a trivial topology, with a standard gradient flow model that evolves along a path of decreasing energy. Domain size is $1 \times 1$. Mesh spacing is $h = 1 / 64$. Row 1: initial configuration. Row 2: $\epsilon = 0.05$. Row 3: $\epsilon = 0.09$. Row 4: $\epsilon = 0.4$. The relaxation occurs from left to right in the second, third and fourth rows. The colour bar represents the value of $\tilde{s} _{\epsilon} ^{2} = \tr{\tilde{\mathbf{Q}} _{\epsilon} ^{2}} / 2$.}
  \label{fig:nontrivial_1}
\end{figure}

\begin{figure}
  \centering
  \includegraphics[width = \textwidth] {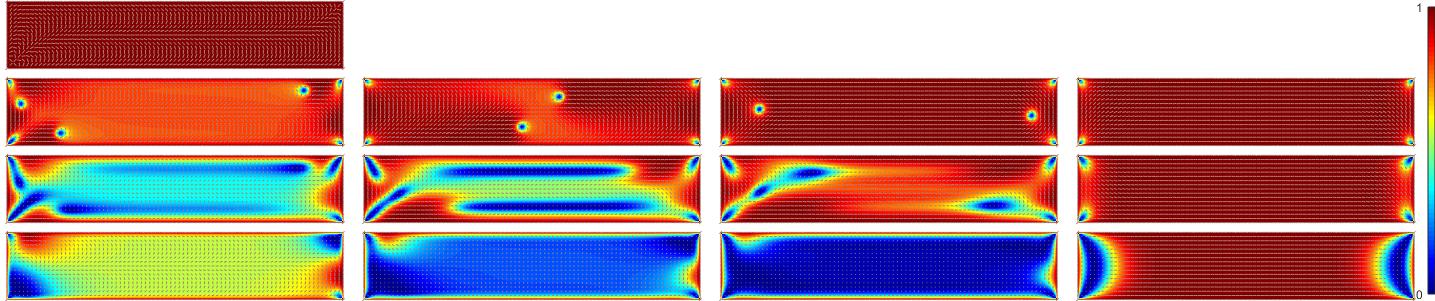} \\
  \caption{Relaxation of a non-trivial topology to a trivial topology, with a standard gradient flow model that evolves along a path of decreasing energy.. Domain size is $5 \times 1$. Mesh spacing is $h = 1 / 64$. Row 1: initial configuration. Row 2: $\epsilon = 0.05$. Row 3: $\epsilon = 0.09$. Row 4: $\epsilon = 0.4$. The relaxation proceeds from left to right in the second, third and fourth rows. The colour bar represents the value of $\tilde{s} _{\epsilon} ^{2} = \tr{\tilde{\mathbf{Q}} _{\epsilon} ^{2}} / 2$.}
  \label{fig:nontrivial_5}
\end{figure}

\section{Conclusion}
\label{sec:conclusions}

We have systematically studied reduced 2D nematic equilibria on rectangles, as a function of a material-dependent and geometry-dependent parameter, $\epsilon$. Our results carry over immediately to the three-dimensional framework at the fixed temperature $A = -\frac{B ^{2}}{3 C}$. Following the methods in \cite{wang2019order}, the qualitative conclusions will hold for all $A < 0$, at least in the limiting cases.

We compare our results on a rectangle with those on a square, with analytic results in two asymptotic limits: the $\epsilon \to \infty$ and the $\epsilon \to 0$ limits respectively. Our first result concerns the loss of the WORS (featured by two isotropic mutually perpendicular defect lines connecting two pairs of diagonally opposite vertices) on a rectangle with Dirichlet or infinitely strong anchoring on the edges. On a rectangle, we get two isotropic defect lines connecting two distinct pairs of adjacent vertices, for nano-scale rectangles in the $\epsilon \to \infty$ limit. Interestingly, the WORS structure survives on squares with relaxed but sufficiently strong anchoring in the $\epsilon \to \infty$ limit, as can be deduced from a uniqueness and symmetry argument. Our arguments are quite generic and we can extend them to arbitrary 2D polygons. In \cite{Han2019reduced}, the authors study reduced 2D LdG equilibria on regular polygons and in the $\epsilon \to \infty$ limit, they report that the limiting profiles have a unique isotropic point at the centre of the regular polygon, with the exception of the square. We can work with irregular polygons too and we speculate that the location of the zeroes or isotropic points are determined by the geometrical anisotropies. In fact, we conjecture that we may only have isotropic interior point defects on generic 2D polygons with tangent boundary conditions in the $\epsilon \to \infty$ limit and the special symmetries of the square and rectangle yield higher-dimensional line defects in this limit. It remains to be seen if we can obtain multiple interior point defects in this distinguished limit by exploiting the geometrical features in 2D.

In the $\epsilon \to 0$ limit for micron-scale or larger geometries, we study the interplay between three competing stable equilibria: the diagonal states and two different types of rotated states: the R1, R2 states and the R3, R4 states. Since the rectangle has lesser symmetry than a square and for $a > b$, the R3 and R4 states have lower energies than the R1, R2 states. We use a simple symmetry based argument to deduce that the diagonal states have lower energy than the competing R3, R4 states, in the limiting scenario. The asymptotic results are complemented by numerical studies of the solution bifurcations on rectangles, as a function of $\epsilon$ and the geometrical anisotropy. We find eight different solutions---the two diagonal solutions, the R1, R2, R3, R4 solutions as enumerated above and the BD1 and BD2 solutions respectively. The BD2 solutions are precisely the limiting profiles in the $\epsilon \to \infty$ limit. Interestingly, the BD solutions with transition layers/nodal (zero) lines near a pair of opposite edges are always unstable on a square. In the case of a rectangle with $a > b$, the BD2 solutions are stable whilst the BD1 solutions are unstable. The geometrical anisotropy essentially disconnects the R1, R2 solutions from the remainder of the bifurcation diagram. The effects of the geometrical anisotropy are more pronounced on the R1, R2 solution pathway in the sense that the structural transitions along this pathway depend on the geometrical parameter $a$. In fact, the BD1 solution (although unstable) is only observed for small values of $a$. This is an interesting example of the effects of $a$, or geometrical anisotropy in our framework.

In the last part, we study the relaxation mechanisms for tangent states with non-trivial topologies into the stable diagonal or rotated states, or the sBD2 state, for different values of $\epsilon$ and the rectangular aspect ratio. In \cite{majumdar2004elastic}, the authors perform related numerical experiments in an Oseen-Frank framework which can be identified with our 2D LdG framework under the restriction of constant $s$; the authors predict that the relaxation mechanisms should either concentrate near the edges or take place in the interior according to the aspect ratio. Our numerical experiments illustrate the effects of $\epsilon$; we find that the relaxation mechanisms are dominated by ejections of bubbles of `non-trivial topology' and the size of these bubbles depends on $\epsilon$. We do not observe clear trends with respect to the aspect ratio possibly because our numerical examples are not exhaustive, the relaxation mechanisms are not unique and we have not been able to locate the aspect-ratio dependent relaxation mechanisms or because the Oseen-Frank studies in \cite{majumdar2004elastic} are  restrictive and lose relevance in the more comprehensive Landau-de Gennes framework. These relaxation mechanisms are of practical relevance for transition pathways between equilibria, switching mechanisms etc. and we will study these mechanisms more carefully in the future. To conclude, our model examples on a rectangle shed useful insight into the consequences of geometrical anisotropy, how anisotropy can distort and displace defects in stable equilibria, stabilise states, disconnect solution branches and in doing so, introduce new possibilities for multi-stability in tailor-made nematic-filled geometries. 

\section{Acknowledgement}

The authors acknowledge support from a Royal Society Newton Advanced Fellowship.
LF acknowledges support from the Visiting Postgraduate Scholar programme at University of Bath.
LF thanks Yucen Han for comments.
LF and LZ were partially supported by National Natural Science Foundations of China (NSFC 11871339, 11861131, 11571314).
AM acknowledges support from Shanghai Jiao Tong University as a Visiting Professor. 
AM was supported by an EPSRC Career Acceleration Fellowship EP/J001686/1 and EP/J001686/2 and is supported by an OCIAM Visiting Fellowship and the Keble Advanced Studies Centre.
The authors thank Giacomo Canevari for very useful discussions on the weak anchoring section.

\section*{Reference}


\bibliography{notes_LiquidCrystals_ref_190517}{}
  \bibliographystyle{alpha}

\end{document}